\documentclass{article}
\usepackage{amsmath, amsfonts, amsthm, amssymb,verbatim}

\usepackage{graphicx,color}
\usepackage{mathrsfs} 
\usepackage{hyperref}
\newtheorem{theorem}{Theorem}[section]
\newtheorem{proposition}[theorem]{Proposition} 
\newtheorem{lemma}[theorem]{Lemma}
\newtheorem{corollary}[theorem]{Corollary} 
\newtheorem{definition}[theorem]{Definition} 
\newtheorem{remark}[theorem]{Remark} 
  
\newcommand \la \langle
\newcommand \ra \rangle

\newcommand \del	\partial

\newcommand \Ecal 	{\mathcal E}

\newcommand \eps 	\epsilon 

\newcommand \be 		{\begin{equation}}
\newcommand \ee 		{\end{equation}}

\newcommand{\tir}{\tilde{r}}

\newcommand{\tio}{\widetilde{\Omega}}

\newcommand{\bir}{\bar{r}}

\newcommand{\bpr}{\bar{\varpi}}

\newcommand{\wtr}{\widehat{\tilde{r}}}
\newcommand{\delt}{\Delta_{\delta,u_0}}

\let\oldmarginpar\marginpar
\renewcommand\marginpar[1]{\-\oldmarginpar[\raggedleft\footnotesize #1]%
{\raggedright\footnotesize #1}}
\author{Gustav Holzegel\footnote{Princeton University, Department of Mathematics, Fine Hall, Washington Road, Princeton, NJ 08544, United States.}
\, and Jacques Smulevici\footnote{
Max-Planck-Institut f\"ur Gravitationsphysik, 
Albert-Einstein Institut, Am M\"uhlenberg 1, 14476 Golm, Germany. }
}
\title{Self-gravitating Klein-Gordon fields in asymptotically Anti-de-Sitter spacetimes}

\begin{document}
\maketitle
\begin{abstract}
We initiate the study of the spherically sym\-metric Eins\-tein-Klein-Gordon system in the presence of a negative cosmological constant, a model appearing frequently in the context of high-energy physics. Due to the lack of glob\-al hyper\-bolicity of the solutions, the natural formulation of dynamics is that of an initial boundary value problem, with boundary conditions imposed at null infinity. We prove a local well-posedness statement for this system, with the time of existence of the solutions depending only on an invariant $H^2$-type norm measuring the size of the Klein-Gordon field on the initial data. The proof requires the introduction of a renormalized system of equations and relies crucially on $r$-weighted estimates for the wave equation on asymptotically AdS spacetimes.
The results provide the basis for our companion paper establishing the global asymptotic stability of Schwarzschild-Anti-de-Sitter within this system.
\end{abstract}

\tableofcontents

\section{Introduction}

The study of local and global well-posedness for linear and nonlinear evolution equations is a traditional subject of mathematical physics. In general relativity, the type of equations range from linear scalar or tensorial field equations on fixed spacetime manifolds to the full non-linear Einstein equations, possibly coupled with matter. 

Whereas a considerable literature is available when the spacetimes under consideration are either asymptotically flat or asymptotically de-Sitter, comparatively few results address the case of asymptotically Anti-de-Sitter (AdS) spacetimes. While the study of field equations on such manifolds certainly deserves mathematical attention in its own right, there is also notable interest from the high energy physics community, see \cite{Julianos,Gubser}. 

The main difficulty to understand the evolution in the case of a negative cosmological constant (and a key difference to both the asymptotically flat and the de Sitter case) is rooted in the lack of global hyperbolicity of the spacetimes one wishes to construct. This fact turns the problem of evolution into an initial-boundary value problem for the Einstein equations. Such problems are intricate in general and a subject of current research (see \cite{Friedrichbv} for a recent survey). Moreover, in the Anti de Sitter case, the boundary is actually located ``at infinity" which causes additional difficulties in the formulation of the dynamics.

\subsection{Wave equations on asymptotically AdS spacetimes}
To gain some intuition into the nature of the problem, one may first study solutions to the linear massive wave equation
\begin{eqnarray} \label{eqphi}
\square_g \phi -\frac{2 a}{l^2} \phi=0
\end{eqnarray}
on a fixed asymptotically AdS spacetime $(\mathcal{M},g)$. Here $l$ is related to the cosmological constant $\Lambda$ as $\Lambda=-3/l^2$ and $a$ is the (squared) Klein-Gordon mass. 
Note that with $a=-1$, \eqref{eqphi} correponds to the conformally invariant wave equation\footnote{Indeed, the conformal wave equation is $\square_{g}\phi-\frac{1}{6}R\phi=0$, where $R$ is the scalar curvature of $g$. For vacuum spacetimes, $R_{\mu\nu}=\Lambda g_{\mu \nu}$ and hence $R= 4\Lambda$, so that the conformal wave equation becomes $\square_{g} \phi + \frac{2}{l^2} \phi=0$.}, which may be considered as the natural analogue of the massless wave equation on asymptotically flat vacuum spacetimes.


Since asymptotically AdS spacetimes are necessarily non-globally hyperbolic, the natural formulation of dynamics for \eqref{eqphi} requires imposing suitable boundary conditions at null infinity. This issue is naturally present in the simplest case, namely that of pure AdS. In this case, existence of solutions for \eqref{eqphi} for a large range of boundary conditions is known (see  \cite{Breitenlohner}, \cite{Bachelot}, \cite{Vasy2}), \cite{Holzegelwp}) if the mass $a$ satisfies the so-called Breitenlohner-Freedmann (BF) bound:
\begin{eqnarray}\label{bd:bf}
a > -9/8.
\end{eqnarray} 
Hence, the value of the mass plays an important role for the well-posedness of this equation. For pure AdS, this can be understood by transforming the equation \eqref{eqphi} to a wave equation on a domain of Minkowski space with a (mass-dependent) potential that becomes singular on a timelike boundary of the domain (see, for instance, the introduction of \cite{Bachelot2}).  More on this in section \ref{prere}.
\subsection{The main result}
In this paper, we shall not be interested in the wave equation \eqref{eqphi} on a \emph{fixed} background but in the non-linearly coupled Einstein-Klein-Gordon system within spherical symmetry. That is to say, we are interested in triples of the form $(\mathcal{M},g,\phi)$, where $(\mathcal{M},g)$ is a $3+1$ Lorentzian manifold, $\phi$ satisfies the Klein-Gordon equation \eqref{eqphi} with respect to $g$, and such that moreover the Einstein equations hold:
\begin{eqnarray} \label{eq:e}
R_{\mu \nu}-\frac{1}{2}g_{\mu \nu}R+\Lambda g_{\mu \nu}=8 \pi T_{\mu \nu},
\end{eqnarray}
where
\begin{eqnarray} \label{def:tmn}
T_{\mu \nu}=\partial_\mu \phi \partial_\nu \phi-\frac{1}{2}g_{\mu \nu} (\partial \phi )^2-\frac{a}{l^2} \phi^2 g_{\mu \nu},
\end{eqnarray}
and $R_{\mu\nu}$, $R$ denote respectively the Ricci tensor and scalar of the metric $g$. Moreover, we assume that $(\mathcal{M},g,\phi)$ is spherically symmetric, i.e.~that there exists a smooth, effective, isometric action of $SO(3)$ on $(\mathcal{M},g)$ leaving invariant both $\phi$ and $g$. Finally, we shall require $(\mathcal{M},g)$ to be asymptotically Anti-de-Sitter\footnote{The precise definitions of asymptotically Anti-de-Sitter spacetimes and of the regularity considered in this paper for $(\mathcal{M},g,\phi)$ are given in Section \ref{se:aads} and \ref{se:fskgf}.}.

The main result of this paper establishes local existence and uniqueness of solutions of \eqref{eqphi}-\eqref{eq:e}-\eqref{def:tmn} for appropriate initial data and boundary conditions, provided the mass satisfies the Breitenlohner-Freedmann bound \eqref{bd:bf}. A concise formulation of our main theorem is therefore:

\begin{theorem} \label{th:mtc}
The system \eqref{eqphi}-\eqref{eq:e}-\eqref{def:tmn}, with $a > -9/8$ and with Dirichlet conditions imposed on $\phi$ at null-infinity, is well-posed for the class of $\mathcal{C}^{1+k}_{a,M}$ asymptotically Anti-de-Sitter data introduced in Definition \ref{def:assbids}.
\end{theorem}
Our Dirichlet conditions imply that the mass is constant along null infinity. A priori, other boundary conditions for $\phi$ could be considered, for instance Neumann boundary conditions. However, for such boundary conditions, the mass flux through null infinity would be infinite. Hence, requiring the mass flux to be finite fixes the boundary conditions of Theorem \ref{th:mtc}. 

Moreover, we establish that the time of existence of the solution, with respect to a bounded null coordinate system, only depends on the value of an invariant $H^2$-type norm for the data of the Klein-Gordon field (and, of course, the choice of coordinates). 
A more precise version of the above theorem is contained in Theorem \ref{wellposed}, with the functional framework for the Klein-Gordon field introduced in Section  \ref{se:fskgf}. 

The data for Theorem \ref{th:mtc} will be prescribed on an outgoing null hypersurface. Hence, we actually prove local well-posedness for the \emph{characteristic} boundary initial value problem. While this emphasizes the geometric character of the problem and, as is well known, simplifies the construction of initial data, we remark that the standard boundary initial value problem can nonetheless be handled by the same method and type of estimates. 

Since away from null infinity, local well-posedness for our system follows by standard techniques, we shall also localize the initial data to a neighboorhood of null infinity, in particular, away from any trapped surface and away from any center of symmetry.



\subsection{Lessons from the linear theory}\label{se:llt}
Before presenting the main ingredients of the proof, let us recall the following insights from the linear theory for solutions in the energy class of \eqref{eqphi} on asymptotically AdS spacetimes:

\begin{itemize}
\item If $a<0$, the dominant energy condition does not hold for the energy momentum tensor of $\phi$ associated with \eqref{eqphi} and hence the natural energy density is not necessarily positive. However, one can use weighted-Hardy inequalities to show that the energy integral is still coercive, with the energy being an $H^1$-type norm on $\phi$.
\item Weighted $H^1$-type norms with radial weights stronger than that present in the natural energy can be propagated by the equations, provided one commutes the equation by a timelike asymptotically Killing field $T$, whose existence is guaranteed by the asymptotics of the metric and the staticity of AdS.
\item In particular, one can establish that $T(\phi)$ has the same radial decay and integrability properties as $\phi$, which in turn leads to improved estimates for some lower order derivatives.
\end{itemize}
In \cite{Holzegelwp} the weighted norms were defined on spacelike hypersurfaces. Here we are going to work with null hypersurfaces and obtain null-versions of the Hardy inequalities.

\subsection{Elements of the proof}
On top of the above ingredients, which are used to control the behaviour of the Klein-Gordon field, the proof of our main result involves:

\begin{itemize}
\item \emph{The introduction of a renormalized system.} Indeed, several geometrical quantities, such as the area-radius function or the conformal factor in appropriate (i.e.~bounded, null) coordinates, blow up at the boundary. However, we remark that (unlike in the vacuum case, cf.~the comments at the end of Section \ref{pres}) even after the renormalization procedure, not all quantities remain finite at the boundary.

\item \emph{Using the Hawking mass as an independent dynamical variable.} This has the advantage that the latter satisfies an easier boundary condition ($\varpi=M$), which in addition is invariant under coordinate changes, something which is not true for the conformal factor, which is typically used. This formulation also allows a resolution for the problem of propagation of constraints: Note that in this characteristic boundary value problem only the $u$-constraint can be propagated from the data, while the validity of the $v$-constraint has to be \emph{established} on the timelike boundary before it can be propagated into the interior.

\item \emph{Control of some higher order derivatives of the metric} in order to make the results of \cite{Holzegelwp} applicable in the context of the contraction map for $\phi$. We remark in this context that our contraction map combines pointwise estimates for the metric components with $L^2$-energy estimates for $\phi$. The reason is that the boundary conditions for $\phi$ do not allow one to integrate directly along characteristics from infinity. 

\item \emph{Control on the difference of solutions to \eqref{eqphi} for two different (but close) asymptotically Anti-de-Sitter metrics $g, g^\prime$} to establish the contraction property. This issue is, of course, not present in the linear case. Its resolution relies crucially on the improved weighted estimate for $T(\phi)$ obtained after commutation. We note that despite the problem being semi-linear,  one can prove the contraction property only in a weaker norm and retrieve the full regularity a posteriori by standard arguments. This feature is normally characteristic of quasi-linear problems and enters here because of the  asymptotically AdS boundary conditions.
\end{itemize}

\subsection{Further results and consequences}
Apart from our main theorem (Theorem \ref{wellposed}), we shall prove several other results useful for further analysis of the system. In particular, we provide an explicit construction of the initial data sets to which our local well-posedness result applies (see Proposition \ref{prop:consdata}). The existence and uniqueness of a maximal solution is established in Corollaries \ref{cor:uniqueness} and \ref{cor:eumd}. 
In Proposition \ref{pro:epin}, we formulate an extension principle applicable near infinity, which is a direct consequence of Theorem  \ref{wellposed}. In the appendix, we formulate a second extension principle, which is an easy adaptation of the recent work\footnote{We thank Mihalis Dafermos and Jonathan Kommemi for pointing out the existence of this extension principle and communicating the work \cite{CamJon}.} \cite{CamJon} to the problem studied here. This second extension principle does not require any form of coercive energy integral arising from the Hawking mass, but lower and upper bounds on the area radius. It is therefore applicable in the interior of the spacetime, away from null infinity, and thus enables us to describe the global structure of the solutions. 


\subsection{Schwarzschild-AdS and the issue of stability}
One special family of solutions to the system studied here is given by the Schwarzschild Anti-de-Sitter spacetimes, 
which solve the system \eqref{eqphi}-\eqref{eq:e}-\eqref{def:tmn}, with $\phi$ being identically $0$. 
With the results of this paper we may consider the maximal solution arising from data which are suitably close to a Schwarzschild-AdS data set, 
and address the issue of stability of Schwarzschild-AdS within the spherically-symmetry Einstein-Klein-Gordon system. This question is resolved in our subsequent paper \cite{gs:stab}, in which we prove global asymptotic stability of the domain of outer communication.

\subsection{Previous results} \label{prere}
\subsubsection*{Linear Theory}
The analysis of the linear problem was initiated in \cite{Breitenlohner}. In \cite{Bachelot}, self-adjoint extensions for the Klein-Gordon operator on pure AdS are constructed for a large class of boundary conditions, if the mass satisfies the BF bound. Moreover, the Dirac system is analysed and in particular, a bound similar to the BF bound is derived. In \cite{Vasy2}, the well-posedness of the linear scalar wave equation on asymptotically Anti-de-Sitter spacetimes admitting a conformal compactification is shown to hold under the $BF$ bound for $a$ and Dirichlet boundary conditions, and a description of the propagation of singularities is given. In \cite{Holzegelwp}, the well-posedness of \eqref{eqphi} is shown for solutions in the energy class (which automatically imposes Dirichlet boundary conditions) using purely vector field techniques. In particular, one has uniqueness in the energy class. Furthermore, in \cite{HolzegelAdS}, boundedness for solutions of \eqref{eqphi} in the energy class is shown for spacetimes which are $C^1$-close to a slowly rotation Kerr-AdS background, again under the assumption that the BF-bound is valid.
 In \cite{Bachelot}, the Klein-Gordon equation in a domain of the five dimensional pure AdS spacetime is analysed, and in particular decay estimates are obtained with respect to a time coordinate adapted to the domain. 

\subsubsection*{Non-linear results}\label{pres}
In the asymptotically flat case, the local and global properties of the spherically symmetric (massless) Einstein scalar field system are well understood, see for instance \cite{Christodoulou, Christodoulou3, Christodoulou4, DafRod}.

In \cite{FriedrichAdS}, the conformal method is used to prove existence and uniqueness of asymptotically AdS spacetimes for the \emph{vacuum} Einstein equations, without any symmetry assumptions. The situation is however quite different from our setting, because there the conformal rescaling provides a complete regularization of the system. In the case of coupling with matter, such as the massive particles of this paper, no such regularization is known to exist.
Finally, we refer to the review article \cite{Friedrichbv} for a general discussion of initial boundary problems for the Einstein equations (in particular, concerning the question of uniqueness).

\subsection{Outline}
The outline of this paper is as follows. In Section \ref{se:pre}, we present the reduced spherically-symmetric Einstein-scalar field equations in double null coordinates and introduce the main geometric quantities needed later. In particular, the precise notion of asymptotically AdS spacetimes used in this paper and the functional framework for the Klein-Gordon field are introduced. 
In Section \ref{se:cids}, we define (and construct) the class of initial data for which our main result will apply. The main theorem is then stated in Section \ref{se:tmt}. In Section \ref{se:rv}, we introduce a renormalization of our system, 
define the associated function spaces and state a local-wellposedness result for this renormalized system. 
Section \ref{se:pmp} is devoted to the proof of this result and contains the key estimates. This allows us to conclude the proof of the main theorem in Section \ref{se:mtp}. In the last section of the paper, we derive some simple consequences of our main result, in particular, the existence of a maximal solution and an extension principle. A second extension principle is given in the appendix. Those results play a key role in our subsequent paper \cite{gs:stab}.


\section{Preliminaries} \label{se:pre}

\subsection[The Einstein-Klein-Gordon equations in null coordinates]{The spherically-symmetric Einstein-Klein-Gordon\\ system in double null coordinates}
We start by recalling a standard result concerning the warped product structure of the metric for spherically symmetric solutions and the form of the equations in double null coordinates (see for instance \cite{Mihali1} and references therein):

\begin{lemma} \label{lem:pre}
Let $(\mathcal{M},g,\phi)$, with $(\mathcal{M},g)$ a $C^2$ Lorentzian manifold, dim $\mathcal{M}=4$ and $\phi$ a $C^2(\mathcal{M})$ function, be a solution to the system \eqref{eqphi}-\eqref{eq:e}-\eqref{def:tmn}. Assume that  $(\mathcal{M},g,\phi)$ is invariant under an effective action of $SO(3)$ with principal orbit type a $2$-sphere. Denote by $r$ the area-radius of the spheres of symmetry. Then, locally around any point of $\mathcal{M}$, there exist double null coordinates $u,v$ such that the metric takes the form:
\begin{eqnarray} \label{eq:metric}
g=-\Omega^2 du dv + r^2 d\sigma_{S^2},
\end{eqnarray}
where $\Omega$ and $r$ may be identified with $C^2$ functions depending only on $(u,v)$ and where $d\sigma_{S^2}$ denotes the standard metric on $S^2$.
Let $\mathcal{Q}=\mathcal{M}/SO(3)$ denote the quotient of the spacetime by the isometry group. 
Then, the Einstein-Klein-Gordon equations\footnote{By a small abuse of notation, we denote functions on $\mathcal{M}$ and their projections to $\mathcal{Q}$ by the same symbols.} reduce to:
\begin{eqnarray}
\partial_u \left( \frac{r_u}{\Omega^2} \right) &=&-4\pi r \frac{(\partial_u \phi)^2}{\Omega^2}, \label{cons1} \\
\partial_v \left( \frac{r_v}{\Omega^2} \right) &=&-4\pi r \frac{(\partial_v \phi)^2}{\Omega^2}, \label{cons2} \\
r_{uv}&=&-\frac{\Omega^2}{4r}-\frac{r_u r_v}{r}+4\pi r (\frac{a \Omega^2 \phi^2}{2l^2})+\frac{1}{4}r \Omega^2 \Lambda, \label{eq:ruv}\\
 \left (\log \Omega \right)_{uv} &=& \frac{\Omega^2}{4r^2}+\frac{r_u r_v }{r^2}-4 \pi \partial_u \phi \partial_v \phi, \\
\partial_u \partial_v \phi  &=& -\frac{r_u}{r} \phi_v - \frac{r_v}{r} \phi_u - \frac{\Omega^2 a}{2l^2} \phi. \label{laste}
\end{eqnarray}
\end{lemma}
Note that the last equation is simply the Klein-Gordon equation \eqref{eqphi} for a spherically symmetric scalar field, since in this case:
$$
0 = \Box_g \phi - \frac{2a\phi}{l^2} = -\frac{4}{\Omega^2} \left( \partial_u \partial_v \phi +\frac{r_u}{r} \phi_v + \frac{r_v}{r} \phi_u \right) - \frac{2a\phi}{l^2}.
$$
Note also that defining $\kappa = -\frac{\Omega^2}{4r_u}$ we have (using (\ref{cons1}))
\begin{align} \label{kappae}
\partial_u \log \kappa &= \frac{4 \pi r}{r_u} (\partial_u \phi )^2 \, .
\end{align}

In view of the lemma, we shall study in the remainder of the article the system \eqref{cons1}-\eqref{laste}. We remark that, while in the statement of the lemma the coordinates $(u,v)$ are only locally defined, the solutions constructed in this paper will always possess global null coordinates. The notation introduced in the above lemma shall be used freely in the following.
\subsection{The Hawking mass}
An important geometric quantity is the renormalized\footnote{The standard definition of the Hawking mass does not include the $\Lambda$ term.} Hawking mass. For any spherically symmetric solution, this is defined as follows:
\begin{eqnarray}
\varpi := \frac{r}{2}\left(1+\frac{4r_u r_v}{ \Omega^2} \right) - \frac{\Lambda}{6} r^3 = \frac{r}{2}\left(1+\frac{4r_u r_v}{ \Omega^2} \right) + \frac{r^3}{2l^2}.
\end{eqnarray}
From \eqref{cons1}-\eqref{laste} it follows that $\varpi$ satisfies
\begin{eqnarray} 
\partial_u \varpi = -8 \pi r^2 \frac{r_v}{\Omega^2} (\partial_u \phi )^2+\frac{4\pi r^2 a }{l^2} r_u \phi^2, \label{ee:varpiu}\\
\partial_v \varpi = -8 \pi r^2 \frac{r_u}{\Omega^2} (\partial_v \phi )^2+\frac{4\pi r^2 a }{l^2} r_v \phi^2. \label{ee:varpiv}
\end{eqnarray}
Note that for $a<0$ the Hawking mass is \emph{not} monotone in the region where $r_u<0$, $r_v\geq 0$ hold. We also introduce the mass ratio
$$
\mu := 1 - 4 \frac{r_u r_v}{\Omega^2} \textrm{ \ \ \ \ \ \ \ and hence \ \ \ \ \ \ \ } 1-\mu=1 -\frac{2 \varpi}{r}+ \frac{r^2}{l^2}.
$$
Similarly, we introduce the quantity $1-\mu_M$ as:
\begin{equation} \label{eq:mum}
1-\mu_M:=1 -\frac{2 M}{r}+ \frac{r^2}{l^2},
\end{equation}
i.e.~$\mu_M$ is the mass ratio of Schwarzschild-AdS.
The wave equation \eqref{eq:ruv} for $r$ may be rewritten using the Hawking mass $\varpi$ as:

\begin{equation}
r_{uv}=-\frac{\Omega^2 \varpi}{2r^2}-\frac{\Omega^2 r}{2 l^2}+\frac{2\pi r a \Omega^2 \phi^2}{l^2} \, .
\end{equation}
\subsection{Triangular domains}
For any real numbers $u_0$, $\delta > 0$, we denote by $\Delta_{\delta,u_0}$ the following triangular subset of $\mathbb{R}^{2}$:
$$\Delta_{\delta,u_0} =\{ (u,v) \in \mathbb{R}^{2} \ \ | \ \  v \ge u_0 \ \ , \ \ u_0+\delta \ge u > v\},$$
depicted below.
\begin{figure}
\[
\input{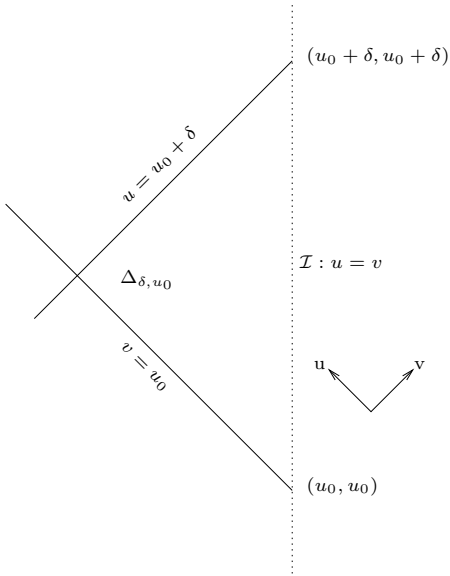}
\]
\caption{The triangular domains $\Delta_{\delta,u_0}$.}
\end{figure}
The restriction of $u$ and $v$ to $\delt$ gives a system of global coordinates on $\delt$ which we will call standard coordinates on $\delt$.
Note that only the $\{v=u\}$ boundary part of $\Delta_{\delta,u_0}$ does not belong to $\Delta_{\delta,u_0}$ itself. This part will be called null-infinity and referred to by $\mathcal{I}$. When we refer to the boundary of $\Delta_{\delta,u_0}$ in the future, we actually mean only the $\{v=u\}$ part of the boundary.

The main objective of this paper will be, for appropriate initial data given on $v=u_0$ and boundary data given on $\mathcal{I}$, to prove the existence of a solution to the system \eqref{cons1}-\eqref{laste} in $\delt$. 
\subsection{Asymptotically Anti-de-Sitter spacetimes} \label{se:aads}
Let $(u,v)$ be standard null coordinates associated to $\delt$ and let $(\Omega, r) \in C^1(\delt) \times C^2(\delt)$ be such that $\Omega > 0$ and $r > 0$ on $\delt$. 
Let $g$ be the Lorentzian metric defined by:
$$
g=-\Omega^2 du dv +r^2 d\sigma_{S^2}.
$$
We say that $(\delt \times S^2,g)$ is asymptotically Anti-de-Sitter if there exists a null coordinate system $(U,V)$ on $\Delta_{\delta,u_0}\cup \mathcal{I}$ such that $U=V$ on $\mathcal{I}$, $r_U < 0$, $r_V > 0$ and $g$ has the following assymptotic behaviour as $\mathcal{I}$ is approached\footnote{Note that while this definition is compatible with having negative AdM mass, the solutions constructed in this paper will all have positive mass $M > 0$. In fact, the mass terms below could be omitted, being of the same order as the $\mathcal{O}$-terms. We included them because later these mass terms will be shown to be the dominant contribution at order $\frac{1}{r}$.}:
\begin{itemize}
\item{ $r(p) \rightarrow \infty$ as $p \rightarrow \mathcal{I}$, }\
\item{ $g = - \left(1 + \mathcal{O}\left(\frac{1}{r^{3}}\right) \right) \left(1-\frac{2M}{r} + \frac{r^2}{l^2} \right) dU dV + r^2  d\sigma_{S^2}$,}
\item{ $R^*(r)= 1-\frac{2M}{r} + \frac{r^2}{l^2}+ \mathcal{O} \left( \frac{1}{r} \right)$,\quad $T(r) = \mathcal{O}\left( \frac{1}{r} \right)$,}
\item{ $T (g_{UV}) = \mathcal{O}\left(\frac{1}{r}\right)$, \quad $R^*( g_{UV} ) = \mathcal{O}\left(r^3\right)$,}
\item{ $R^*(r_U)=\mathcal{O}\left(r^3\right)$, \quad $R^*(r_V)=\mathcal{O}\left(r^3\right)$,}
\item{ $T(r_U)=\mathcal{O}\left(\frac{1}{r}\right)$, \quad $T(r_V)=\mathcal{O}\left(\frac{1}{r}\right)$,}
\end{itemize}
where $R^*=\partial_V-\partial_U$ and $T=\partial_V+\partial_U$.

We note that this definition is compatible with that given in \cite{Holzegelwp}. Indeed, define the coordinates $(t,r^\star)$ by $U = t- r^\star$, $V = t+r^\star$ and observe that $r_{r^\star} = \Omega^2+ \mathcal{O}\left(\frac{1}{r}\right)$, while $r_t = \mathcal{O}\left(\frac{1}{r}\right)$. This means that
\begin{align}
\frac{1}{\Omega^2} dr = \left(1 + \mathcal{O}\left(\frac{1}{r^3}\right)\right) dr^\star + \mathcal{O}\left(\frac{1}{r^3}\right) dt
\end{align}
and hence that the metric $g$ can be written in $\left(t,r\right)$ coordinates as
\begin{align}
g = -\left(1 - \frac{2M}{r} +\frac{r^2}{l^2}+ \mathcal{O}\left(\frac{1}{r}\right) \right) dt^2 + \left(\frac{1}{1+\frac{r^2}{l^2} - \frac{2M}{r}} + \mathcal{O}\left(\frac{1}{r^5}\right) \right) dr^2 \nonumber \\ +
\left( \mathcal{O}\left(\frac{1}{r^3}\right) \right) dt \ dr + r^2 d\sigma_{S^2},
\end{align}
which is compatible with the asymptotics of the metric of an asymptotically AdS spacetime introduced in \cite{Holzegelwp}. Since moreover, the asymptotic behaviour of the derivatives of $g$ also agrees with that of \cite{Holzegelwp}, $(\delt,g)$ is asymptotically AdS in the sense of \cite{Holzegelwp}.
\subsection{Function spaces for the Klein-Gordon field} \label{se:fskgf}
We shall use weighted Sobolev spaces to control the regularity of the Klein-Gordon-field. These norms are motivated by the radial weights arising in the energy estimates for asymptotically-AdS spacetimes.


Let us hence introduce the following weight function (one should think of this as being of the order of the inverse area radius), defined on any interval of the form $\mathcal{N}=(u_0,u_1]$:
$$
\bar{\rho} = \frac{u-u_0}{2}.
$$
We then defined the weighted $H^1$ norm, $||.||_{H^1_{AdS}(\mathcal{N})}$, as: 
\begin{eqnarray}
|| \psi ||_{H^1_{AdS}(\mathcal{N})}= \left( \int_{\mathcal{N}}\left( (\bar{\rho})^{-2} \bar{\psi}_u^2+ (\bar{\rho})^{-4} \psi^2 \right)du' \right)^{1/2}.
\end{eqnarray}
Note that we have the following $\bar{\rho}$-weighted Sobolev inequality:
$$
|| (\bar{\rho})^{-3/2} \psi ||_{C^0\left(\mathcal{N}\right)} \le  \sqrt{\frac{2}{3}} ||\psi||_{H^1_{AdS}\left(\mathcal{N}\right)},
$$
Similarly, we define on the entire $\delt$, the weight factor:
$$\rho(u,v)=\frac{u-v}{2}$$
and define $|| .||_{C^0\left(H^1_{AdS}\right)}$ for functions $\phi$ defined on $\delt$, as the following weighted Sobolev type spacetime norm:
\begin{eqnarray}
|| \psi ||_{C^0\left(H^1_{AdS}\right)}&=& 
\hbox{}\sup_{(u,v)\in \delt}\left( \int_{v}^{u}\left( \rho^{-2} \psi_u^2 + \rho^{-4}\psi^2 \right)(u',v)du'\right)^{1/2} \nonumber \\
&&\hbox{}+\sup_{(u,v) \in \delt} \left( \int_{u_0}^{v} \left( \rho^{-2} \psi_v^2+\rho^{-4} \psi^2 \right)(u,v')dv'\right)^{1/2} \nonumber
\end{eqnarray}

The completion of the set of compactly supported $C^\infty$ functions with respect to the above norms defines a Banach space, which we call $C^0\left(H^1_{AdS}\right)$. We then define $C^{1}\left(H^1_{AdS}\right)$ to be the set of functions $\phi \in C^0\left(H^1_{AdS}\right)$ such that $T(\phi)=\phi_u+\phi_v$ lies also in $C^0\left(H^1_{AdS}\right)$, and we endow $C^{1}\left(H^1_{AdS}\right)$ with the norm 
$$
||\phi||_{C^{1}\left(H^1_{AdS}\right)}:=||\phi||_{C^{0}\left(H^1_{AdS}\right)}+ ||T(\phi)||_{C^{0}\left(H^1_{AdS}\right)}.
$$
While the $\rho$-weights appearing in the $H^1_{AdS}$ norms take their origin in the standard energy estimate associated to the problem, stronger $\rho$-weighted estimates may be propagated by the equations, using a commutation argument (see in particular \cite{Holzegelwp} on this issue). The extra $\rho$-weights that one can gain depend on the value of the mass $a$ and more precisely on how far one is from saturating the BF bound. To measure this extra radial decay, it will be convenient to introduce the constant
\begin{align} \label{sdef}
s = \min \left(\sqrt{9+8a},1\right)\, ,
\end{align}
In view of the BF bound \eqref{bd:bf}, one has $s > 0$. In Proposition \ref{lineartheory}, pointwise decay rates\footnote{Note that in \cite{Holzegelwp}, stronger decay rates are actually established, provided they hold initially. Those will however not be needed for the purpose of this paper.} will be established, with the rate of decay depending on the value of $s$.
%
%
%
%
%
%
%
%
\subsection{Notation}
To establish the results of this paper, we will introduce a renormalized version of the system (see Section \ref{se:rv}). In order for the reader to track the geometric meaning of the renormalized quantities, it will be useful to introduce the following notation:

\begin{itemize}
\item{\textbf{Renormalized variables}:

All renormalized variables will be referred to by an upper tilde. For instance, the renormalized radial function will be $\tilde{r}$ and the renormalized conformal factor $\widetilde{\Omega}$.}
\item{\textbf{Initial data}:

All variables variables belonging to the initial data will be refered to by an upper bar. For instance, $\bar{\tilde{r}}$ denotes the initial data for the renormalized radial function, while $\bar{r}$ denote the initial data for the original radial function.
}
\end{itemize}
\section{Asymptotically AdS characteristic data} \label{se:cids}

\subsection{Construction of asymptotically AdS data sets}
We now turn to the definition and construction of asymptotically AdS data sets for the Einstein-Klein-Gordon equations within spherical symmetry.

Looking at (\ref{cons1})-(\ref{laste}) it seems that we need to prescribe initial values for $r$, $\phi$ and $\Omega$ on $v=u_0$ (satisfying the constraint (\ref{cons1})), and then use the equations to determine $r_v$, $\phi_v$ and $\Omega_v$. We choose a slightly different approach using the Hawking mass as the dynamical variable.

To understand the content of the following Proposition \ref{prop:consdata} note first that prescribing the geometric area radius on $\mathcal{N}$ is equivalent to specifying a $u$-coordinate along $\mathcal{N}$. Second, specifying $\phi$ corresponds to the ``free-data" (since we know that for $\phi=0$ the solution of the evolution problem is expected to be Schwarzschild-AdS, in view of Birkhoff's theorem). One then observes that $\varpi$ can be determined from $\phi$ and $r$ using the constraint equation. From this, all other variables are also determined. Since we wish $\phi$ and $T(\phi)$ to live in the energy space, we shall need ensure that both $\phi$ and $T(\phi)$ satisfy appropriate decay assumptions. 

\begin{proposition} \label{prop:consdata}
Let $k \geq 1$, $a > -9/8$, $M > 0$ and let $\mathcal{N}$ be a real interval of the form $\mathcal{N}=(u_0,u_1]$.
Let $(\bar{r},\bar{\phi} ) \in C^{k+1}(\mathcal{N}) \times C^{k+1}(\mathcal{N})$ with the following properties:
\begin{itemize}
\item{ Monotonicity and sign properties of $\bar{r}$:
\begin{eqnarray}
\bar{r} &>& 0, \label{rsign}\\
\bar{r}_u &<& 0. \label{rusign}
\end{eqnarray}}
\item { Asymptotic behaviour and choice of $u$-coordinates: as $u \rightarrow u_0$,
\begin{eqnarray} 
\bar{r}(u) &\rightarrow& \infty, \label{eq:ras} \\
2\bar{r}_u(u)&=&-\left(1-\overline{\mu_M}\right)(\bar{r}(u))+o\left(\bar{r}^{-1}\right),  \label{eq:ruas0}\\
\partial_u \left(\frac{\bar{r}_{u}}{1-\overline{\mu_M}\left(\bar{r}\right)}\right) &=& o\left( \bar{r}^{-2}\right), \label{eq:ruas}
\end{eqnarray}
where $(1-\overline{\mu_M})(\bar{r})=1-\frac{2M}{\bar{r}}+\frac{\bar{r}^2}{l^2}$ (cf.~\eqref{eq:mum}).
}

\item{Radial decay of the Klein-Gordon field: }
\begin{align} \label{as:dec}
|r^{ \frac{s}{2}} \bar{\phi} | + \Big| \bar{r}^{1+\frac{s}{2}}\frac{\bar{\phi}_u}{\bar{r}_u}\Big| + \frac{\bar{r}^{2}}{\bar{r}_u} \left( \frac{\bar{\phi}_{u}}{\bar{r}_u} \right)_u \leq \mathcal{O} \left( \bar{r}^{-\frac{3}{2}} \right),
\end{align}
where $s$ is the constant introduced in \eqref{sdef}.
\item{Further integrability conditions:}

Definining $\bar{\Phi}$ as:
$$\bar{\Phi}=\bar{r}^2 \left[ \bar{r} \partial_u \left( \frac{ \bar{\phi}_u } { \bar{r}_u} \right)- 4 \bar{\phi}_u - \frac{2a \bar{r}_u}{\bar{r}} \bar{\phi} \right],$$ we have the following integrability properties:
\begin{eqnarray} 
\bar{\Phi} &\in& L^1(\mathcal{N}),  \label{as:s1dec}\\
\bar{r}_u \Pi^2 &\in& L^1(\mathcal{N}), \label{as:s2dec}
\end{eqnarray}
where $\Pi(u)=\int_{u_0}^u \bar{\Phi}(u')du'$.

\end{itemize}
Then, there exists a unique triple $(\bar{\varpi}, \overline{r_v}, \overline{\phi_v}) \in C^k(\mathcal{N})^3$ such that $\left(\bar{r}, \overline{r_v}, \bar{\phi}, \overline{\phi_v}, \bar{\varpi}\right)$ satisfies the following conditions: 

\begin{itemize}
\item $r_v$ sign: $ \overline{r_v}>0$ holds for all $u$ sufficiently close to $u_0$. Also, the quantity
$$\bar{\kappa} : = \frac{\overline{r_v}}{1-\overline{\mu}} $$ 
is positive and satisfies
\begin{equation} \label{bc:kappa}
\lim_{u\rightarrow u_0} \bar{\kappa} = \frac{1}{2} \, .
\end{equation}
%

\item{Asymptotic behaviour of the Hawking mass:}
\begin{equation}\label{varpias}
\lim_{u\rightarrow u_0} \bar{\varpi}\left(u\right) = M \, .
\end{equation}
\item $H^1$ bound for the time derivative of $\phi$:
Defining $\overline{\mathcal{T}}(\overline{\phi})$ as
\begin{eqnarray}
\overline{\mathcal{T}}(\overline{\phi}):= \frac{1}{\bar{\kappa}} \left( \overline{\phi_v} - \frac{\overline{r_v}}{\bar{r}_u} \bar{\phi}_u \right),
\end{eqnarray}
we have the $H^1$ bound: 
\begin{align} \label{tinen}
\int_{\mathcal{N}} \bar{r}^2 \left(\frac{\bar{r}^2}{\bar{r}_u} \left|\frac{d}{du} (\overline{\mathcal{T}}(\overline{\phi}))\right|^2+ \overline{\mathcal{T}}(\overline{\phi})^2 \bar{r}_u  \right) du' < \infty.
\end{align}

\item{Finally, the following constraint equations are satisfied:}

\textbf{- Hawking mass constraint:}
\begin{equation} \label{vpd}
\partial_u \bar{\varpi} = 8 \pi \bar{r}^2 \frac{1-\frac{2\bar{\varpi}}{\bar{r}} + \frac{\bar{r}^2}{l^2}}{4\bar{r}_u} \left(\partial_u \bar{\phi}\right)^2 + \frac{4\pi \bar{r}^2 a}{l^2} \bar{r}_u \bar{\phi}^2,
\end{equation}
\textbf{- $\overline{r_v}$ equation:}
$$
\left( \overline{r_v} \right)_u = \left(- \frac{\bpr}{2\bir^2}-\frac{\bir}{2l^2}+\frac{2\pi a}{l^2} \bir \bar{\phi}^2\right)\left( - \frac{4 \bir_u \overline{r_v}}{1-\frac{2\bpr}{\bir}+\frac{\bir^2}{l^2}}\right),
$$

\textbf{ - $\overline{\phi_v}$ equation:}
\begin{align} \label{wep}
\partial_u \left( \overline{\phi_v} \right) + \frac{\bar{r}_u}{\bar{r}}\overline{\phi_v}+ \frac{\overline{r_v}}{\bar{r}}\bar{\phi}_u= -\frac{\bar{\Omega}^2 a }{2l^2} \bar{\phi}.
\end{align}
\end{itemize}
\end{proposition}
\begin{proof}
Define the quantity $\bar{\varpi}$ as the unique solution of \eqref{vpd}
with boundary condition given by \eqref{varpias}. Having constructed $\bar{\varpi}$, we define $\overline{r_v}$ as the $C^k$ quantity (cf.~(\ref{kappae})) :
\begin{align} \label{rvk}
\overline{r_v} = \frac{1}{2} \left( 1 -\frac{2\bar{\varpi }}{\bar{r}}+\frac{\bar{r}^2}{l^2} \right) \exp \left( \int^u_{u_0} \frac{4 \pi \bar{r}}{\bar{r}_u} \left( \partial_u \bar{\phi} \right)^2 du' \right),
\end{align} 
using equations \eqref{eq:ruas0} and \eqref{as:dec}.
We also define the $C^k$ quantities $\bar{\Omega} > 0$ and $\bar{\kappa}$ by
\begin{eqnarray}
\bar{\Omega}^2&=&-\frac{4\bar{r}_u\overline{{r}_v}}{1-\frac{2\bar{\varpi}}{\bar{r}}+\frac{\bar{r}^2}{l^2}}, \\
\bar{\kappa} &=& \frac{\overline{r_v} }{1-\frac{2\bar{\varpi}}{\bar{r}}+\frac{\bar{r}^2}{l^2}}.
\end{eqnarray}
\eqref{bc:kappa} is then immediate from (\ref{rvk}). Next we define the $C^1$ quantity $\overline{\mathcal{T}}(\overline{\phi})$ by integrating:
$$
 \frac{d}{du} \left( \bar{r} \bar{\kappa} \overline{\mathcal{T}}(\overline{\phi}) \right)= -\bar{r} \overline{r_v} \frac{d}{du} \frac{\bar{\phi}_u}{\bar{r}_u}+\bar{\phi}_u \left[ -2 \overline{r_v} - 2 \frac{\bar{\kappa} \bar{r}^2}{l^2}-\frac{2 \bar{\kappa} \bar{\varpi}}{\bar{r}}+ \frac{8 \pi \bar{r}^2 a \bar{\kappa} \bar{\phi}^2} {l^2} \right]- \frac{a \bar{\Omega}^2 \bar{r} }{2l^2} \bar{\phi}
$$
with the boundary condition $\bar{r} \bar{\kappa} \overline{\mathcal{T}}(\overline{\phi})=0$ at $u=u_0$. It follows from the decay assumptions \eqref{as:dec} and the conditions \eqref{as:s1dec}-\eqref{as:s2dec} that: 
$$
\int_{\mathcal{N}} \bar{r}^2 \left( \frac{\bar{r}^2}{\bar{r}_u} \left|\frac{d}{du} (\overline{\mathcal{T}}(\overline{\phi}))\right|^2+ \overline{\mathcal{T}}(\overline{\phi})^2 \bar{r}_u  \right) du' < \infty \, .
$$
Finally, $\overline{\phi_v}$ is defined from $\overline{\mathcal{T}}(\overline{\phi})$:
$$
\overline{\phi_v} = \bar{\kappa} \overline{\mathcal{T}}(\overline{\phi})+ \frac{\overline{r_v}}{\bar{r}_u} \bar{\phi}_u.
$$
One can then check that $\left(\bar{r}, \bar{\phi}, \overline{r_v}, \overline{\phi_v}, \bar{\varpi}\right)$ satisfies all the requirements stated in the proposition. The uniqueness of $\overline{r_v}$ follows from the ordinary differential equation satisfied by $\kappa$ and the boundary condition \eqref{bc:kappa}. The uniqueness of $\overline{\phi_v}$ follows from that of $\overline{\mathcal{T}}(\bar{\phi})$, since the $H^1$ bound \eqref{tinen} imposes the boundary conditions $\bar{r} \bar{\kappa} \overline{\mathcal{T}}(\overline{\phi})=0$.
\end{proof}

\begin{remark}
The monotonicity properties and the asymptotics assumed on $\bar{r}\left(u\right)$, $\overline{r_v}$ and $\bar{\kappa}$ correspond to the choice of an asymptotically AdS coordinate system. The decay properties for $\bar{\phi}$ contain the $s$-improvement which is familiar from the linear case.
The integrability conditions  (\ref{as:s1dec})  and (\ref{as:s2dec}) are imposed to ensure that $\overline{\mathcal{T}}\left(\bar{\phi}\right)$ lives in the energy space, (\ref{tinen}). Finally, we introduced the expression $\overline{\mathcal{T}}\left(\bar{\phi}\right)$ because it is manifestly invariant under a change of $u$-coordinate. Alternatively, one could work with the algebraically simpler but non-invariant $\overline{T} \left(\bar{\phi}\right)= \bar{\phi}_u + \overline{\phi_v}$ (cf.~Lemma \ref{le:nr}).
\end{remark}

Proposition \ref{prop:consdata} leads us naturally to the following definition of asymptotically Anti-de-Sitter data sets:
\begin{definition} \label{def:assbids}
Let $k \geq 1$, $a > -9/8$, $M > 0$ and let $\mathcal{N}$ be a real interval of the form $\mathcal{N}=(u_0,u_1]$. A \underline{$\mathcal{C}^{1+k}_{a,M}(\mathcal{N})$ asymptotically Anti-de-Sitter data set} 
is a pair of functions $(\bar{r},\bar{\phi})$ such that there exists a coordinate system on $\mathcal{N}$ in which $\bar{r}$ satisfies \eqref{rsign}-\eqref{rusign} and \eqref{eq:ras}-\eqref{eq:ruas} and $\bar{\phi}$ satisfies \eqref{as:dec}, \eqref{as:s1dec}, \eqref{as:s2dec}. 
\end{definition}
\begin{remark}
In view of the above proposition, we shall sometimes refer, by a small abuse of notation, to $\left(\bar{r}, \overline{r_v}, \bar{\phi}, \overline{\phi_v}, \bar{\varpi}\right)$, where $(\overline{r_v},  \overline{\phi_v}, \bar{\varpi})$ has been constructed from $(\bar{r},\bar{\phi})$, as an asymptotically Anti-de-Sitter data set.
\end{remark}
\subsection{A geometric norm for the initial data} \label{nosec}
The construction of the data suggests to introduce the following norm for the Klein-Gordon field: Given $\left(\bar{r},\bar{\phi}\right)$ an asymptotically AdS data set, we define $||\bar{\phi}||_{H^1_{AdS}\left(\bar{r}, \mathcal{N}\right)}$ as:

$$
|| \bar{\phi}||_{H^1_{AdS}\left(\bar{r},\mathcal{N}\right)}= \left( \int_{\mathcal{N}}\bar{r}^2 \left(  \frac{\bar{r}^2}{|\bar{r}_u|} \bar{\phi}_u^2+ |\bar{r}_u|\bar{\phi}^2 \right)du' \right)^{1/2}
$$
and then the total invariant norm of $(\bar{r},\bar{\phi})$, $N_{inv}\left[ \bar{r},\bar{\phi}\right]$, as:
\begin{eqnarray} \label{nor:inv}
N_{inv}\left[ \bar{r},\bar{\phi}, \mathcal{N}\right]=|| \bar{\phi} ||_{H^1_{AdS}(\bar{r}, \mathcal{N})}+||\overline{\mathcal{T}}(\overline{\phi})||_{H^1_{AdS}(\bar{r},\mathcal{N})}+||\bar{r}^{5/2+s} \frac{\bar{\phi}_u}{\bar{r}_u} ||_{C^0(\mathcal{N})}, \nonumber 
\end{eqnarray}
where $\overline{\mathcal{T}}(\bar{\phi})$ is defined as in Proposition \ref{prop:consdata} and $s$ is the constant introduced in \eqref{sdef} . Note that this norm is invariant under a change of $u$-coordinate. Hence in view of Birkhoff's theorem, it is a geometric measure of the distance beetween our initial data set and an initial data set for Schwarzschild-AdS.  

\begin{remark}
The $H^1_{AdS}$-type norms originate from the energy estimate for the wave equation on asymptotically AdS spacetimes. The pointwise norm on $\bar{\phi}_u$ in $N_{inv}\left[ \bar{r},\bar{\phi},, \mathcal{N}\right]$, on the other hand, is specific to spherical symmetry and exhibits an additional $r$-weight characteristic of the problem, cf.~(\ref{sdef}).
\end{remark}
Let us also define, for any subset $\mathcal{N}^\prime \subset \mathcal{N}$ the quantities
\begin{align}
\mathcal{A} \left[\mathcal{N}^\prime\right]=  \sup_{\mathcal{N}^\prime, 1-\bar{\mu} \neq 0} \Big| \frac{\bar{r}^3}{Ml^2} \left( \frac{\bar{r}_u}{1-\bar{\mu}} + \frac{1}{2} \right)\Big| +   \sup_{\mathcal{N}^\prime, 1-\bar{\mu} \neq 0} \Big| \frac{\bar{r}^2}{l^2} \left( \partial_u \frac{\bar{r}_u}{1-\bar{\mu}} \right)\Big| 
\end{align}
\begin{align}
\mathcal{B}  \left[\mathcal{N}^\prime\right]= \min\left( \inf_{{\mathcal{N}^\prime}} \Big| \frac{1-\bar{\mu}}{\bar{r}^2} \Big| , \inf_{{\mathcal{N}^\prime}} \Big| \frac{1-\overline{\mu_M}}{\bar{r}^2} \Big| \right)
\end{align}
The quantity $\mathcal{A} \left[\mathcal{N}^\prime\right]$ measures how close the choice of $u$-coordinate is to the $u$-Eddington-Finkelstein coordinate of Schwarzschild-AdS. The quantity $\mathcal{B} \left[\mathcal{N}^\prime\right]$ has geometric significance: A lower bound on $\mathcal{B}$ guarantees the absence of trapped surfaces.
Of course, we could have $\mathcal{B} \left[\mathcal{N}\right]=0$ and hence possibly $\mathcal{A} \left[\mathcal{N}\right]=\infty$. However, in view of the asymptotically AdS property, given any $0<c<l^2$, one can always restrict $\mathcal{N}$ to a subset $\mathcal{N}^\prime \subset \mathcal{N}$ near infinity such that $\mathcal{B} \left[\mathcal{N}^\prime\right]>c$. 
Hence, given a $\mathcal{C}^{1+k}_{a,M}(\mathcal{N})$ asymptotically AdS data set $\left(\bar{r}, \overline{r_v}, \bar{\phi}, \bar{\phi}_v, \bar{\varpi}\right)$ we can decompose
\begin{align} \label{decomp}
\mathcal{N} = \mathcal{N}_1 \cup \mathcal{N}_2 = \left(u_0,u_{max}\right] \cup  \left[u_{max},u_1\right] \, ,
\end{align}
where $u_{max}>u_0$ is the supremum over all $u$ such that 
$\mathcal{B}\left[\mathcal{N}_1\right] \geq \frac{l^2}{2}$ holds. We remark that instead of $\frac{l^2}{2}$ one may work with any $0<c<l^2$. Clearly, it then follows that $\mathcal{A}\left[\mathcal{N}_1\right] < \infty$, with the bound depending only on the precise choice of coordinates.

The point of the above definitions is that on $\mathcal{N}_1$, the geometric norm on the Klein-Gordon field is equivalent to the energy norm introduced in Section \ref{se:fskgf}:
\begin{lemma} \label{le:nr}
We have the following estimates on $\mathcal{N}_1$:
\begin{align} \label{pio1}
 || \bar{\phi}||_{H^1_{AdS}\left(\mathcal{N}_1\right)} \leq \tilde{C} || \bar{\phi}||_{H^1_{AdS}\left(\bar{r},\mathcal{N}_1\right)} \, ,
\end{align}
\begin{align} \label{pio2}
 || \overline{T}\bar{\phi}=\bar{\phi}_u+\overline{\phi_v}||_{H^1_{AdS}\left(\mathcal{N}_1\right)} \leq \tilde{C} \left[ || \overline{\mathcal{T}} \bar{\phi}||_{H^1_{AdS}\left(\bar{r},\mathcal{N}_1\right)} + || \bar{\phi}||_{H^1_{AdS}\left(\bar{r},\mathcal{N}_1\right)} \right] \, ,
\end{align}
with $\tilde{C}$ depending only on $\mathcal{A}\left[\mathcal{N}_1\right]$ and $N_{inv}\left[ \bar{r},\bar{\phi},\mathcal{N}_1\right]$ and, such that moreover, $\tilde{C}$ is a decreasing function of $N_{inv}\left[ \bar{r},\bar{\phi},\mathcal{N}_1\right]$.
\end{lemma}
\begin{proof}
From the construction of the initial data, in particular from (\ref{rvk}), one obtains  that
\begin{align} \label{hebon}
\Big| \log \left(\bar{r} \rho \right)\Big|  + \Big| \frac{\bar{r}^3}{Ml^2} \left( \frac{\overline{r_v}}{1-\bar{\mu}} - \frac{1}{2} \right)\Big| + \Big| \frac{\bar{r}^2}{l^2} \left( \partial_u \frac{\overline{r_v}}{1-\bar{\mu}} \right)\Big| < \tilde{C}
\end{align}
holds everywhere on $\mathcal{N}_1$, with $\tilde{C}$ having the dependence as stated in the theorem. This establishes the equivalence of the weights appearing in the norms of (\ref{pio1}) and hence the equivalence of the norms themselves. For (\ref{pio2}), one notices the relation
\begin{align}
\overline{\mathcal{T}} \left(\phi\right) = - \frac{1-\bar{\mu}}{\bar{r}_u} \overline{T} \left(\phi\right) + \overline{\phi_v} \left(\frac{1-\bar{\mu}}{\overline{r_v}} +\frac{1-\bar{\mu}}{\bar{r}_u} \right) \, ,
\end{align}
and that the bracket decays strongly in $\bar{r}$ in view of of the bound on $\mathcal{A}[\mathcal{N}_1]$ and (\ref{hebon}). The estimate (\ref{pio2}) then follows by straightforward computation, 
writing $\overline{\phi_v} = \overline{T}\left(\bar{\phi}\right)-\bar{\phi}_u$.
\end{proof}
%
%
%
%
%
%
%
%
%
%
%
\section{The main theorem} \label{se:tmt}
We are now ready to state the main result of this paper. The following theorem guarantees the existence and uniqueness of a solution in a small triangle localized near null-infinity (in particular, we will restrict the data to $\mathcal{N}_1$). From this one easily obtains a solution in an entire small strip to the future of $\mathcal{N}$, Corollary \ref{prop:ile}.

\begin{theorem} \label{wellposed}
Given a $\mathcal{C}^{1+k}_{a,M}(\mathcal{N})$ asymptotically AdS data set $\left(\bar{r}, \overline{r_v}, \bar{\phi}, \bar{\phi}_v, \bar{\varpi}\right)$ set on $\mathcal{N}=(u_0,u_1]$, there exists a $0<\delta<u_{max}-u_0$ (with $u_{max}$ defined in (\ref{decomp})) such that the following statement is true. There exists a unique solution $\left(r, \Omega,\phi,\varpi\right)$ of the equations (\ref{cons1})-(\ref{laste}) in $\delt$ with
\begin{itemize}
\item $r\left(u,v\right)$ is $C^{1+k} \left(\Delta_{\delta,u_0}\right)$, $\Omega\left(u,v\right)$ is $C^k\left(\Delta_{\delta,u_0}\right)$, 
\item $\phi\left(u,v\right)$ is in $C^{k} \left(\Delta_{\delta,u_0}\right)\cap C^1\left(H^{1}_{AdS}\left(\Delta_{\delta,u_0}\right)\right)$, 
\item and $\varpi\left(u,v\right)$ is in $C^k\left(\Delta_{\delta,u_0}\right)$,
\end{itemize}
and such that on $v=u_0$: $(r,r_v,\phi,\phi_v,\varpi)=(\bar{r},\overline{r_v},\bar{\phi},\overline{\phi_v},\bar{\varpi})$. Moreover, 
\begin{enumerate}
\item the size of the domain of definition of the solution, $\delta$, depends only on 
$\mathcal{A}\left[\mathcal{N}_1\right]$ and the coordinate-invariant norm\footnote{Recall that $\mathcal{N}_1$ was defined in \eqref{decomp}, independently of $\delta$.} $N_{inv}\left[ \bar{r},\bar{\phi}, \mathcal{N}_1\right]$ on the initial data (all defined in Section \ref{nosec}). 
\item the associated spacetime $\left(\Delta_{\delta,u_0} \times S^2, g = -\Omega^2 \left(u,v\right) du dv + r^2\left(u,v\right) d\sigma_{S^2}\right)$ is asymptotically AdS.
\item The trace of the solution on any $v$=const ray contained in $\Delta_{\delta,u_0}$ defines an asymptotically AdS initial data set.
\end{enumerate}
\end{theorem}
\vspace{.3cm}
\begin{remark} \label{rere}
The statement on the size of $\delta$ can be paraphrased by saying that the time of existence depends only on the size of the Klein-Gordon field $\phi$ and the choice of coordinates. 
\end{remark}

The uniqueness statement is to be understood in the following sense:  Any two solutions $\left(r_1, \Omega_1,\phi_1,\varpi_1\right)$ and $\left(r_2, \Omega_2,\phi_2,\varpi_2\right)$ living in the spaces of Theorem \ref{wellposed} and satisfying the equations as well as the initial and boundary conditions have to agree. As in the linear case, imposing that $\phi$ lives in the energy space $H^1_{AdS}$ is crucial for the uniqueness. We can upgrade this uniqueness result to a local geometric uniqueness statement within spherical-symmetry:

\begin{corollary} \label{cor:uniqueness}
Let $\mathcal{D}=(\bar{r},\bar{\phi})$ be an asymptotically AdS data set $\mathcal{C}^{1+k}_{a,M}(\mathcal{N})$. Let $(\mathcal{M}_i,g_i,\phi_i)$, $i=1,2$ be two developments of $\mathcal{D}$. Then, both $(\mathcal{M}_i,g_i,\phi_i)$ are extensions of a common development. 
\end{corollary}
The proof of this corollary, as well as the precise definition of ``development'' (including the regularity)  and ``extension'' of solutions, are to be found in Section \ref{subse:md}. In particular, our definition of development includes an appropriate replacement of global hyperbolicity, as well as the requirement that $\phi$ lives in the energy space. From the above uniqueness statement, one can infer by standard methods \cite{Geroch} the existence of a maximum developement:

\begin{corollary} \label{cor:eumd}
Any asymptotically AdS data set $\mathcal{C}^{1+k}_{a,M}(\mathcal{N})$ admits a maximal development. This development is unique up to isometry.
\end{corollary}

Finally, away from null infinity, we can use standard arguments to obtain the following existence result:
\begin{corollary} \label{prop:ile}
Given a $\mathcal{C}^{1+k}_{a,M}$ asymptotically AdS boundary initial data set on $\mathcal{N}=(u_0,u_1]$, there exists a unique solution of the equations (\ref{cons1})-(\ref{laste}) satisfying the initial and boundary conditions in a thin strip $\left(u_0, u_1\right] \times \left[u_0,u_0+\delta\right] \cap\{ v\leq u \}$.
\end{corollary}
\begin{proof}
Applying Theorem \ref{wellposed} we obtain the existence of a solution in a small triangle of size $\delta$ near $\mathcal{I}$. In particular, all quantities ($r, \phi, \Omega, \varpi$) and their derivatives are all bounded in terms of the data on the outgoing ray $\{u_0+\delta\} \times [u_0,u_0+\frac{\delta}{2}]$. We now pose the characteristic problem with data on $[u_0+\delta, u_1) \times \{u_0\}$ and  $\{u_0+\delta\} \times [u_0,u_0+\frac{\delta}{2}]$. Existence and uniqueness of the solution in a small strip follows by standard estimates since $r$ is bounded above and below on the data (see for instance Proposition \ref{sexp}).
\end{proof}

Theorem \ref{wellposed} will be established by transforming the system of equations to a renormalized system and proving a well-posedness statement for the latter (Proposition \ref{renowp}). Solutions to the renormalized system are then shown to be in one-to-one correspondence with solutions to the original system (Proposition \ref{eq}). We note in this context that the contraction map will only provide the regularity stated in the first part of the Theorem. The improved regularity needed to prove item $3$ (in particular, the pointwise bound on $\phi_{uu}$ assumed in the construction of the data) will be established \emph{after} the existence of a solution has been shown. Finally, let us remark that once a solution is known from Theorem \ref{wellposed}, Theorem 6.1 of \cite{Holzegelwp} applies, from which it follows that $\phi$ is in fact a $C^0\left(H^{2,s}_{AdS}\right)$  function (see \cite{Holzegelwp} for a definition of the space $H^{2,s}_{AdS})$. \newline

\section{The renormalized system} \label{se:rv}
In section \ref{se:cids}, we introduced a class of asymptotically Anti-de-Sitter initial boundary data sets. Unfortunately, several quantities introduced there are blowing up at the boundary, see for instance the asymptotic behaviour of $\bar{r}_u$ and $\overline{r_v}$. In this section, we will introduce a set of renormalized initial boundary data on $\mathcal{N}_1$, which are in one-to-one correspondence with the original data on $\mathcal{N}_1$, but which are better behaved at the boundary (at the cost of losing their original geometric significance). Similarly, we introduce renormalized variables associated to any solution of our system, which are then shown to be in one-to-one correspondence with solutions to the original system (see Proposition \ref{eq}). 
%
%
%
\subsection{Renormalized initial data sets}
The following definition is a simple rewriting of the initial data.
\begin{definition}[Initial data in renormalized variables] \label{lem:idrnv}
Let $\left(\bar{r}, \overline{r_v}, \bar{\phi}, \overline{\phi_v}, \bar{\varpi}\right)$ be a $\mathcal{C}^{1+k}_{a,M}(\mathcal{N})$ asymptotically AdS data set  as in Definition \ref{def:assbids}. Restrict to $\mathcal{C}^{1+k}_{a,M}(\mathcal{N}_1)$ according to (\ref{decomp}). Let $\bar{\tilde{r}}$ be obtained by integrating
\begin{align}
\bar{\tilde{r}}_u = -\frac{\overline{r}_u}{1-\frac{2M}{\bar{r}} + \frac{\bar{r}^2}{l^2}}  ,
\end{align}
with boundary conditions  $\bar{\tilde{r}}\left(u_0\right) = 0$. Let $\overline{\tilde{r}_v}$ be defined as
$$
\overline{\tilde{r}_v} := -\frac{\overline{{r}_v}}{1-\frac{2M}{\bar{r}} + \frac{\bar{r}^2}{l^2}} .
$$
Let $\bar{\tilde{\Omega}}^2 = \frac{\bar{\Omega}^2}{1-\frac{2M}{\bar{r}} + \frac{\bar{r}^2}{l^2}}$, where $\bar{\Omega}$ is as in Proposition \ref{prop:consdata}.

Then, we call $\left(\tilde{\bar{r}}, \overline{\tilde{r}_v}, \bar{\tilde{\Omega}}, \bar{\phi}, \overline{\phi_v}, \bar{\varpi}\right)$ a $\mathcal{C}^{1+k}_{a,M}(\mathcal{N}_1)$ \underline{renormalized data set}. 
\end{definition}

As an immediate consequence of the definition, we note the following facts:

\begin{lemma} Let $\left(\tilde{\bar{r}}, \overline{\tilde{r}_v}, \bar{\tilde{\Omega}}, \bar{\phi}, \overline{\phi_v}, \bar{\varpi}\right)$ be a $\mathcal{C}^{1+k}_{a,M}(\mathcal{N}_1)$ renormalized data set arising from $\left(\bar{r},\bar{\phi}\right)$. 
Then, the equations 
\begin{eqnarray}
\partial_u \left(\frac{\bar{\tilde{r}}_u}{\bar{\tilde{\Omega}}^2} \right)&=&- 4\pi \frac{\bar{r}}{\left(1-\bar{\mu}_M\right)} \frac{\left(\partial_u \bar{\phi}\right)^2}{\bar{\tilde{\Omega}}^2}, \\
1-\bar{\mu} &=& -\frac{4 \bar{\tilde{r}}_u \overline{\tilde{r}_v} \left(1-\overline{\mu_M}\right)}{\bar{\tilde{\Omega}}^2},\\
\left(\overline{\tilde{r}_{v}}\right)_u &=& -\bar{\tio}^2\left[ \frac{2\pi \bar{r} a \bar{\phi}^2}{l^2}+\frac{M-\bar{\varpi}}{2 \left(1-\overline{\mu_M}\right) \bar{r}^2}\left(1+ \frac{3\bar{r}^2}{l^2}\right)\right], \label{ee:tiruv}
\end{eqnarray}
where $(1-\overline{\mu_M})=1-\frac{2M}{\bar{r}}+\frac{\bar{r}^2}{l^2}$, hold in $\mathcal{N}_1$ for the renormalized data set.
\end{lemma}

On top of the renormalization described above, it will be convenient to further localize the data and the solutions to a neighbourhood of $\mathcal{I}$. 
\begin{lemma} \label{lem:btd}
For any $\delta' > 0$, there exists a $\delta>0$ sufficiently small so that the following bounds hold on $\mathcal{N}^\prime=\left(u_0,u_0+\delta\right] \subset \mathcal{N}_1 \subset \mathcal{N}$:
\begin{align}
| \bar{\tilde{r}} - \frac{u-u_0}{2} | + \left|\log (2\bar{\tir}_u) \right| + \left|\log (2\overline{\tir_v}) \right| + |\left(u-u_0\right)^{-1} \left(\tilde{r}_u+\overline{\tir_v}\right) | \le \delta' \nonumber \\
 | \bar{\tir}_{uu}| + | \left(u-u_0\right)^{-2} \left(\overline{\tilde{r}_{v}}\right)_u|   \le \delta',\nonumber \\
 | \log \frac{1-\overline{\mu_M}}{1-\bar{\mu}}| + | \overline{\varpi} - M| + | \left(u-u_0\right) \overline{\varpi}_u | \le \delta'
\nonumber 
\end{align}
Moreover, $\delta$ depends only on $\mathcal{A}\left[\mathcal{N}_1\right]$ and the invariant norm $N_{inv}\left[\bar{r},\bar{\phi}, \mathcal{N}_1\right]$ on the data.
\end{lemma}
\begin{proof}
From the evolution equation \eqref{vpd}, the bounds on $\mathcal{B}[\mathcal{N}_1]$ and $\mathcal{A}[\mathcal{N}_1]$ and the pointwise bounds on $\phi$ and $\phi_u$ (which depend only on $N_{inv}\left[\bar{r},\bar{\phi}, \mathcal{N}_1\right]$), one obtains the estimates on $\varpi$ and $\varpi_u$, choosing $\delta$ suficiently small. The estimates on $\tir$, $\tir_u$ and $\tir_{uu}$  then follow easily from the definition of $\mathcal{A}[\mathcal{N}_1]$. For $\left(\overline{\tilde{r}_{v}}\right)_u$, we use the wave equation \eqref{ee:tiruv}. For $\tilde{r}_u+\overline{\tir_v}$, we use the fact that  $\tilde{r}_u+\overline{\tir_v}=0$ at $u=u_0$ and that the derivative is uniformly bounded in view of the bounds on $\tir_{uu}$ and $\left(\overline{\tir_v}\right)_u$.
\end{proof}

\subsection{Well-posedness for the renormalized system} \label{se:mp}
We are now ready to state the local well-posedness result for the renormalized system:
\begin{proposition} \label{renowp}
Given a renormalized initial data set $\left(\bar{\tilde{r}}, \bar{\tilde{\Omega}}, \varpi, \bar{\phi}\right)$ on $\mathcal{N}_1$ there exists a $0<\delta<u_{max}$ such that there is a unique solution $\left(\tilde{r},\tilde{\Omega}, \varpi, \phi\right) \in \mathcal{C}^2 \left(\overline{\Delta_{\delta,u_0}}\right) \times \mathcal{C}^1 \left(\Delta_{\delta,u_0}\right) \times \mathcal{C}^1 \left(\Delta_{\delta,u_0}\right) \times \mathcal{C}^1 \left(H^1_{AdS}\right)$  
of the following system of equations in the triangle $\Delta_{\delta,u_0}$:
\begin{eqnarray}
\tilde{r}_{uv} &=& -\tio^2\left[ \frac{2\pi r a \phi^2}{l^2}+\frac{M-\varpi}{2 \left(1-{\mu}_M\right) r^2}\left(1+ \frac{3r^2}{l^2}\right)\right],\nonumber \\
\left(\log \tilde{\Omega}^2 \right)_{uv}  &=& -8 \pi \partial_u \phi \partial_v \phi  + \frac{\tilde{\Omega}^2}{r^3} \left(\varpi-M\right), \nonumber\\
&&\hbox{} - \left(\frac{2M}{r^2} + \frac{2r}{l^2}\right) \tio^2\left[ \frac{2\pi r a \phi^2}{l^2}+\frac{M-\varpi}{2 \left(1-\mu_M\right) r^2}\left(1+ \frac{3r^2}{l^2}\right)\right],
 \nonumber\\
\varpi_u &=& -8\pi r^2 \frac{-\tir_{v}}{\tio^2} \left( \partial_u \phi \right)^2 + \frac{4\pi a}{l^2}  r^2 \tir_u(1-\mu_M) \phi^2,  \nonumber\\
\Box_g \phi &=&  \frac{2a}{l^2} \phi, \nonumber
\end{eqnarray}
where $\Box_g$ is the wave operator associated with the metric $g = \frac{4 \tilde{r}_u \tilde{r}_v}{1-\mu} \left(1-\bar{\mu}\right)^2 du \ dv + r^2 d\sigma_{S^2}$, and where $r$ is a strictly positive $C^2$ function satisfying:
\begin{eqnarray}
r_u&=&- \tir_u (1-\mu_M), \label{eq:rtiru} \\
r &\rightarrow& \infty, 
\end{eqnarray}
such that the solution restricts on $v=u_0$ to the prescribed data and the boundary conditions $\varpi \rightarrow M$, $\tilde{r} \rightarrow 0$ and $-4\frac{\tilde{r}_u \tilde{r}_v}{\tilde{\Omega}^2} \rightarrow 1$ hold.\footnote{Recall that the Dirichlet boundary condition on $\phi$ is automatic by membership in $ \mathcal{C}^1 \left(H^1_{AdS}\right)$.} \newline
Finally, $\delta$ depends only on $\mathcal{A}\left[\mathcal{N}_1\right]$ and the norm $N_{inv}\left[\bar{r},\bar{\phi}, \mathcal{N}_1\right]$ of the initial data  (all defined in Section \ref{nosec}) .
\end{proposition}
The proof of this proposition is the subject of Section \ref{se:pmp}.
As a corollary, we obtain the propagation of the constraints:

\begin{corollary} \label{cor:progconst}
Under the assumptions of Proposition \ref{renowp}, the equations
\begin{align}
\partial_u \left(\frac{\tilde{r}_u}{\tilde{\Omega}^2} \right)- 4\pi \frac{r}{\left(1-\mu_M\right)} \frac{\left(\partial_u \phi\right)^2}{\tilde{\Omega}^2}=0 \ ,  \label{fih} \\
\partial_v \left(\frac{\tilde{r}_v}{\tilde{\Omega}^2}\right) - 4\pi \frac{r}{\left(1-\mu_M\right)} \frac{\left(\partial_v \phi\right)^2}{\tilde{\Omega}^2}=0  \ , \label{seh}
\end{align}
as well as
\begin{equation} \label{thh}
1-\mu = \frac{4 \tilde{r}_u \tilde{r}_v \left(1-\mu_M\right)}{\tilde{\Omega}^2}
\end{equation}
hold in $\Delta_{\delta,u_0}$.
\end{corollary}
\begin{proof}
Note that $\tilde{\Omega}_{uv}$ is $C^0$.
The equation $\partial_u \left(\frac{\tilde{r}_u}{\tilde{\Omega}^2}\right) - 4\pi \frac{r}{\left(1-\mu_M\right)} \frac{\left(\partial_u \phi\right)^2}{\tilde{\Omega}^2}=0$ holds on the data ray $v=u_0$ by construction. With the regularity established, we can differentiate the expression in $v$.  Let $A= \partial_u \left(\frac{\tilde{r}_u}{\tilde{\Omega}^2}\right)$ and $B= 4\pi r \frac{\phi_u^2}{\tilde{\Omega}^2 \left(1-\mu_M\right)}$. One computes
\begin{align}
\partial_v A = - 2 \frac{\tilde{\Omega}_v}{\tilde{\Omega}} A + 8 \pi \frac{\tilde{r}_u}{\tilde{\Omega}^2} \phi_u \phi_v - \frac{4\pi r a}{l^2} \phi \phi_u + \left(\frac{1}{r^2} + \frac{3}{l^2} \right) \frac{4\pi r^2 \tilde{r}_v}{\tilde{\Omega}^2 \left(1-\mu_M\right)} \phi_u^2 \nonumber \\
+ \frac{a \phi^2 \tilde{r}_u }{l^2} \left[ -2\pi \left(1 + \frac{3r^2}{l^2}\right) + 2\pi \left(1-\mu_M\right) + 4\pi \left(\frac{M}{r} + \frac{r^2}{l^2} \right) \right] \nonumber \\
 = - 2 \frac{\tilde{\Omega}_v}{\tilde{\Omega}} A + 8 \pi \frac{\tilde{r}_u}{\tilde{\Omega}^2} \phi_u \phi_v - \frac{4\pi r a}{l^2} \phi \phi_u + \left(\frac{1}{r^2} + \frac{3}{l^2} \right) \frac{4\pi r^2 \tilde{r}_v}{\tilde{\Omega}^2 \left(1-\mu_M\right)} \phi_u^2
\end{align}
and
\begin{align}
\partial_v B =  - 2 \frac{\tilde{\Omega}_v}{\tilde{\Omega}} B + 4\pi \phi_u^2 \frac{\tilde{r}_v}{\left(1-\mu_M \right) \tilde{\Omega}^2}  \left( \left(1-\mu_M\right)  +  \left(\frac{2M}{r} + \frac{2r^2}{l^2} \right) \right) \nonumber \\
+ 8 \pi \frac{\tilde{r}_u}{\tilde{\Omega}^2} \phi_u \phi_v  - \frac{4\pi r a}{l^2} \phi \phi_u
\end{align}
and hence
\begin{align}
\partial_v \left(A-B\right) = - 2 \frac{\tilde{\Omega}_v}{\tilde{\Omega}} \left(A-B\right) \, .
\end{align}
We conclude that $A-B=0$ everywhere in the triangle as it holds initially, establishing (\ref{fih}). To prove (\ref{thh}), note that
the equation $ 4\frac{{\tilde{r}}_u \tilde{r}_v}{\tilde{\Omega}^2} \frac{\left(1-\mu_M\right)}{1-\mu}+1 =0$ holds on the boundary $u=v$. We can differentiate in $u$ 
\begin{align}
\partial_u \left(4\frac{{\tilde{r}}_u {\tilde{r}}_v}{\tilde{\Omega}^2} \frac{\left(1-\mu_M\right)}{1-\mu}+1 \right) = \left(4\frac{{\tilde{r}}_u {\tilde{r}}_v}{\tilde{\Omega}^2} \frac{\left(1-\mu_M\right)}{1-\mu}+1 \right) \left[-\frac{2\varpi_u}{\left(1-{\mu}\right)r} \right]
\end{align}
and conclude that $ 4\frac{{\tilde{r}}_u{\tilde{r}}_v}{\tilde{\Omega}^2} \frac{\left(1-\mu_M\right)}{1-\mu}+1 =0$ holds everywhere.

Finally, for (\ref{seh}), we observe first that $\partial_u \left(\frac{\tilde{r}_u}{\tilde{\Omega}^2}\right) = 0$ on the boundary. However, since also $\tilde{r}_{uv}=0$ on the boundary we have in fact $\partial_u \left(\frac{\tilde{r}_u \tilde{r}_v }{\tilde{\Omega}^2}\right) = 0$ but the expression in brackets is also constant along the boundary by construction, which means that actually $\partial_v \left(\frac{\tilde{r}_u \tilde{r}_v }{\tilde{\Omega}^2}\right) = 0$ also. This in turn means that $\partial_v \left(\frac{\tilde{r}_v}{\tilde{\Omega}^2}\right) = 0$ and hence $\partial_v \left(\frac{\tilde{r}_v}{\tilde{\Omega}^2}\right) - 4\pi \frac{r}{\left(1-\mu_M\right)} \frac{\left(\partial_v \phi\right)^2}{\tilde{\Omega}^2}=0$ holds on the boundary. Differentiating in $u$ we see that this identity is propagated into the triangle, the computation being entirely analogous to that of the $u$-constraint above.
\end{proof}
\section{Proof of Proposition \ref{renowp}}\label{se:pmp}
The proof is based on the construction of a contracting map. We start by introducing an appropriate metric space for the renormalized variables.
\subsection{Function spaces for the renormalized variables}
Recall the weight $\rho\left(u,v\right)=\frac{u-v}{2}$ and $T=\partial_u + \partial_v$.
We denote by $C^{2,uv}_{\tir}(\delt)$ the set of $C^{2}(\delt)$ positive functions $\tir \left(u,v\right) > 0$ which satisfy
\begin{eqnarray} \label{bc:tirf}
\tir_{uv} \in C^1(\delt) , \ \ \ \ \ \frac{1}{2} \le \frac{\tir}{\rho} \le 2,  \, 
 \nonumber \\ 
|\tilde{r}_u - \frac{1}{2}| \leq \frac{1}{4}, \ \ \ \ \ |\tilde{r}_v + \frac{1}{2}| \leq \frac{1}{4} \,.\nonumber
\end{eqnarray}
On this space of functions, we define the following distance\footnote{In this definition, the $\log$ could have been omitted for the derivatives of $\tilde{r}$, in view of the bounds on $\tilde{r}_u$ and $\tilde{r}_v$. It is included only for computational convenience.}
\begin{eqnarray}
d_{\tir}( \tilde{r}_1, \tilde{r}_2 )&=& || \log \frac{\tir_1}{\tir_2} ||_{C^0}
+ || \rho^{-1} \left[ T\left(\tilde{r}_1\right) - T\left(\tilde{r}_2\right)\right]||_{C^0} \nonumber \\
&&\hbox{}+ || \log [(\tir_1)_u] -  \log [(\tir_2)_u] ||_{C^0}+|| \log [(-\tir_1)_v] -  \log [-(\tir_2)_v] ||_{C^0} \nonumber \\
&&\hbox{}+ \| \rho^{-2} \left[\left(\tilde{r}_1\right)_{uv} - \left(\tilde{r}_2\right)_{uv} \right] \|_{C^0} +|| \rho^{-2}\left[T\left(\tir_1 \right)_{uv}-T\left( \tir_2 \right)_{uv}\right] ||_{C^0} \nonumber \\
&&\hbox{}+ \|  \left(\tilde{r}_1\right)_{uu} - \left(\tilde{r}_2\right)_{uu} \|_{C^0} + \| \left(\tilde{r}_1\right)_{vv} - \left(\tilde{r}_2\right)_{vv} \|_{C^0} \, . \nonumber
\end{eqnarray}

Let $C^{1}_{\tio}(\Delta)$ denote the set of $C^1$ functions $\tio$ in $\delt$ which are bounded below by $1/2$. On this space, we define the following norm:
\begin{eqnarray}
d_{\tio}(\tio_1,\tio_2)&=&\Big\| \log \left(\tio_1\right)^2 -  \log \left(\tio_2\right)^2 \Big\|_{C^0} \nonumber \\ 
&&\hbox{} \| \left(\tio_1\right)_{u} - \left(\tio_2\right)_{u} \|_{C^0} + \| \left(\tio_1\right)_{v} - \left(\tio_2\right)_{v} \|_{C^0}  \, .
\end{eqnarray}
We denote by $C^1_{\varpi}(\Delta)$ the set of $C^1$ functions $\varpi \left(u,v\right)$ on $\delt$
equipped with the weighted $C^1$ norm
\begin{align}
d_{\varpi}( \varpi_1, \varpi_2 )&=  \| \rho^{-s/4}\left(\varpi_1 - \varpi_2 \right) \|_{C^0} + \| \rho^{-s/8}T \left(\varpi_1 - \varpi_2 \right) \|_{C^0} \nonumber \\ &+  \| \rho\left(\varpi_1 - \varpi_2 \right)_u \|_{C^0} +  \| \rho \left(\varpi_1 - \varpi_2 \right)_v \|_{C^0},   \, 
\end{align}
with $s$ defined in (\ref{sdef}).\footnote{Stronger $\rho$-weighted estimates may be propagated by the equations. In particular, one can show boundedness of $\rho^{\min(2,\sqrt{9+8a})/2-\epsilon} \varpi$ for any $\epsilon>0$, cf.~\cite{gs:stab}.} The additional $\rho$-weights will be one of the sources of smallness for the contraction map, the other arising from the size of the domain of definition, $\delta$.

For the Klein-Gordon field we shall use the norm $||.||_{C^1 \left( H^{1}_{AdS} \right)}$ introduced in Section \ref{se:fskgf} and a pointwise norm on $\phi_u$,
\begin{align} \label{pwpu}
|| \phi ||_{\mathring{C}^{\frac{2+s}{4}}_{u}(\delt)} := || (\rho)^{-1/2-\frac{s}{4}}  \phi_u||_{C^0(\delt)} \, ,
\end{align}
with $s$ defined in (\ref{sdef}). Note that, as for the Hawking mass, stronger $\rho$-weighted estimates may in fact be propagated by the equations.
The space of functions used for $\phi$ will then be $\mathring{C}^{\frac{2+s}{4}}_{u}(\delt)  \cap  C^1\left( H^{1}_{AdS} \right)$.
Finally, we define the complete metric space $\mathcal{C}$ by 
$$\mathcal{C}=C^{2,uv}_{\tir} \times C^1_{\tio} \times C^1_{\varpi} \times \left(\mathring{C}^{\frac{2+s}{4}}_{u}(\delt)  \cap  C^1\left( H^{1}_{AdS} \right) \right),$$
endowed with the distance $d$:
\begin{align}
d\left(( \tir_1,\varpi_1,\tio_1,\phi_1 ),( \tir_2,\varpi_2,\tio_2,\phi_2 )\right)= d_{\tir}\left( \tir_1, \tir_2 \right)+d_{\tio}\left( \tio_1, \tio_2 \right) + d_{\varpi}(\varpi_1,\varpi_2)\nonumber \\
+ ||\phi_1-\phi_2||_{C^1(H^{1}_{AdS})} + ||\phi_1-\phi_2||_{\mathring{C}^{\frac{2+s}{4}}_{u}(\delt)}\nonumber
\end{align}
and denote by $B_{\mathcal{C},b}^{\delt}$ the closed ball of radius $b$ centered around $$\left( \frac{u-v}{2},1,M,0\right).$$
Note the trivial fact that if $u \in B_{\mathcal{C},b}^{\delt}$ and if $0 < \delta' \le \delta$, then $$u|_{\Delta_{\delta',u_0}} \in B_{\mathcal{C},b}^{\Delta_{\delta',u_0}}.$$
For this reason, we shall also use the notation $B_{\mathcal{C},b}$ for the ball, without explicit reference to the triangular domain.
\subsection{Properties of the elements of $B_{\mathcal{C},b}$}
Before constructing the contraction map, it will be useful to establish some properties associated with elements of $B_{\mathcal{C},b}$. First we shall show that from any element of $B_{\mathcal{C},b}$ one can reconstruct the area-radius function $r$ (i.e.~such that $\tir$ is the renormalized variable associated to $r$: $\tilde{r}_u = \frac{r_u}{1-\mu_M}$). This is more easily done by constructing first $r^{-1}$, as this inverse quantity remains finite at infinity:
\begin{lemma} \label{lem:asarv}
Consider an element $\left(\tilde{r},{\tilde{\Omega}},\varpi,{\phi}\right) \in B_{\mathcal{C},b}$.
Let  $f\left(u,v\right)$ be the unique solution, for each fixed $v$, of 
\begin{align}
f_u(u,v) &= {\tir}_u(u,v)(f^2- 2M f^3+\frac{1}{l^2})(u,v)  \label{fiftyo} \\
f\left(v,v\right)& =0 \nonumber \, .
\end{align}
and define $r\left(u,v\right)=f^{-1}(u,v)$.
If $\delta$ is sufficiently small, depending only on $M,l$ (but independent of $b$), we have the following estimates:
\begin{eqnarray} \label{es:tirrinver}
|r^{-1}| \le C \tir  \textrm{ \ \ \ and  \ \ \ } |r| \le  C {\tir}^{-1} 
\end{eqnarray}
where $C>0$ only depends on $M$, $l$.
%
Moreover, we have the equations
\begin{eqnarray} \label{eq:rrt}
r_u=-\tir_u (1-\mu_M), \quad r_v=-\tir_v (1-\mu_M).
\end{eqnarray}

\end{lemma}
\begin{proof}
Let us first show that $f > 0$ in $\delt$. Indeed, this holds near $u=v$, in view of  $f(v,v)=0$ and $f_u(v,v) > 0$. Assume that there exists some $v_1$ such that $f(.,v_1)$ vanishes at some point in $\delt$, and let $u_1$ denote the first $u$ such that $f(u_1,v_1)$, hence  at $(u_1,v_1)$, we must have $f_u \leq 0$. However, $f_u > 0$ at $(u_1,v_1)$ by (\ref{fiftyo}), a contradiction. 
Since $f > 0$, we have the estimate
$$
f_u \le \tilde{r}_u \left( f^2 + \frac{1}{l^2} \right),
$$
from which it follows that:
$$
0 < \arctan (l f) \le \frac{1}{l} \tilde{r}.
$$
Since $\frac{1}{2}\le \frac{\tilde{r}}{\rho} \le 2$, it then follows that 
$f \le \frac{2}{l^2} \tilde{r}$ if $\delta$ is sufficiently small depending only on $l$.

We may then estimate $f^3$ by $\tilde{r}^3$.
Hence, we may also obtain a lower bound on $f$ of the form: 
$$
f \ge C \tilde{r}- C\tilde{r}^4 \ge  C \tilde{r},
$$
where $C$ only depends on $M$ and $l$ and provided $\delta$ is chosen sufficiently small depending only on $M$ and $l$.
%
To derive the equations \eqref{eq:rrt}, let us consider the (differentiable, monotonically decreasing) function
\begin{eqnarray}
F: (r_{Sch},\infty) &\rightarrow& \mathbb{R} \nonumber \\
   x &\mapsto& F(x)=\int_x^\infty \frac{dy}{1-\frac{2M}{y}+\frac{y^2}{l^2}}, \nonumber
\end{eqnarray}
where $r_{Sch}$ is the unique real root of $P(y)=1-\frac{2M}{y}+\frac{y^2}{l^2}$. It is easy to see that $r(u,v)=F^{-1}\left( \tir(u,v) \right)$, from which equations \eqref{eq:rrt} follow by the chain rule.
\end{proof}

We next observe that we can associate an asymptotically AdS spacetime with any element $\left(\tilde{r},{\tilde{\Omega}},\varpi,{\phi}\right) \in B_{\mathcal{C},b}$:

\begin{lemma} \label{coads}
Consider an element $\left(\tilde{r},{\tilde{\Omega}},\varpi,{\phi}\right) \in B_{\mathcal{C},b}$ and
let  $r\left(u,v\right)$ be defined from $\tilde{r}$ as in Lemma \ref{lem:asarv}. Then the spacetime $\left(\Delta_{\delta,u_0} \times S^2, g\right)$ with 
\begin{align} \label{metric:ebcb}
g= \frac{4 \tilde{r}_u \tilde{r}_v}{1-\mu} \left(1-\mu_M\right)^2 \left(u,v\right) du dv + r^2 \left(u,v\right) d\sigma_{S^2}
\end{align}
is asymptotically AdS. 
\end{lemma}
\begin{proof}
The function $\tilde{r}$ extends to a $C^2$ function on $\overline{\delt}$ which is  vanishing on $\mathcal{I}$. Moreover, since the derivatives of $\tilde{r}_u$ and $\tilde{r}_v$ are uniformly bounded, these functions extend continuously to the boundary. Since $\tilde{r}=0$ is constant on the boundary, we must have $\tilde{r}_u + \tilde{r}_v = 0$ on the boundary. Let us denote the restriction of $\tilde{r}_u$ to the boundary by $g\left(u\right) = \tilde{r}_u \left(u,u\right)$. Similarly, $h\left(v\right) = \tilde{r}_v \left(v,v\right)$. Note in particular that $g$ and $h$ are bounded above and below. Defining the regular double-null coordinate transformation $\frac{dU}{du} = \frac{1}{2}g\left(u\right)$ and $\frac{dV}{dv} = \frac{1}{2} h\left(v\right)$ we see that in the $\left(U,V\right)$ coordinate system we have $\tilde{r}_U = \frac{1}{2}$ and $\tilde{r}_V = -\frac{1}{2}$ on the boundary (which is at $U=V$). Since $\tilde{r}_{uv}=\frac{1}{4} g\left(u\right) h \left(v\right) \tilde{r}_{UV}$, we have also $\tilde{r}_{UV} = \mathcal{O}\left(\tilde{r}^{2}\right)$ and hence that $|\tilde{r}_U - \frac{1}{2}| = \mathcal{O}\left(\tilde{r}^{3}\right)$, $|\tilde{r}_V + \frac{1}{2}| = \mathcal{O}\left(\tilde{r}^{3}\right)$.

In summary, we find that in the new coordinates, the metric reads
\begin{align}
g = - \left(1 + \mathcal{O}\left(\tilde{r}^{3}\right) \right) \left(1-\frac{2M}{r} + \frac{r^2}{l^2} \right) dU dV + r^2 \left(U,V\right) d\sigma_{S^2}
\end{align}
with $2\tilde{r}_U=1$ and $2\tilde{r}_V=-1$ on the boundary. Moreover, since $\tilde{r}_U$ is constant on $\mathcal{I}$, we have $T(\tilde{r}_U)=0$. Consequently,
\begin{eqnarray}
T(\tilde{r}_U)=\int^U_V \left[ \tilde{r}_{UV}+\tilde{r}_{UU} \right]_V dV' \nonumber \\
=\int^U_V \left[ \tilde{r}_{UVV}+\tilde{r}_{UVU} \right] dV'=\int^U_V \left[ T \left(\tilde{r}_{UV}\right)  \right] dV'
 \le C_{b} \tilde{r}^{3},
\end{eqnarray}
and we obtain that $T(\tilde{r}_U)=\mathcal{O} \left( \frac{1}{r^3} \right)$. The same holds for $T(\tilde{r}_V)$. This implies that $T(r_V)=\mathcal{O}\left( \frac{1}{r} \right)$ and $T(r_U)=\mathcal{O}\left( \frac{1}{r} \right)$ and hence that $T(g_{UV})= \mathcal{O}\left( \frac{1}{r} \right)$.
Finally, since $\tilde{r}_{UU}$ and $\tilde{r}_{VV}$ are bounded, we have $r_{UU}=\mathcal{O}(r^3)$ and $r_{VV}=\mathcal{O}(r^3)$. 
Hence, the spacetime is asymptotically AdS in the sense introduced in Section \ref{se:aads}.
\end{proof}
\begin{remark} \label{coc}
From the above proof, we also infer that whatever double-null coordinate system we started from, there always exists a $C^2$ bounded (the bound depending only on the size of the ball, $b$) coordinate transformation to an asympto\-tically-AdS coordinate system $\left(U,V\right)$, in which $\tilde{r}_U = -\tilde{r}_V = \frac{1}{2}$ on the boundary. In particular, any quantity which is uniformly bounded in one coordinate system, is also uniformly bounded in the other.
\end{remark}
The spacetimes constructed from elements of $B_{\mathcal{C},b}$ in Lemma \ref{coads} are \emph{uniformly} asymptotically AdS spacetimes. That is to say that one has uniform pointwise bounds (depending only on $M$, $l$ and $b$) on the metric components and their derivatives. 
Adapting the results of \cite{Holzegelwp} to our spherically symmetric setting, we hence infer the following
\begin{proposition} \label{lineartheory}
Let $\left(\tilde{r},{\tilde{\Omega}},\varpi,{\phi}\right)$ be an element of $\in B_{\mathcal{C},b}$ and consider the spherically symmetric asymptotically AdS spacetime $\left(\Delta_{\delta,u_0},g\right)$ arising from Lemma \ref{coads}. Then the wave equation $\Box_g \psi - \frac{2a}{l^2} \psi =0$ with $H^2$-initial data $(\bar{\phi},\overline{T}(\bar{\phi}))$ 
has a unique solution in $C^1\left( H^1_{AdS} \right)$. For sufficiently small $\delta$, depending only on $b$, the following energy estimates hold for any $\left(u,v\right)$ in a triangle $\Delta_{\delta,u_0}$:
\begin{eqnarray} 
\int_{v}^u \left[ \left(\partial_u T \phi\right)^2 + \left(\partial_u \phi\right)^2 + \phi^2 (1-\bar{\mu}) + \left(T\phi\right)^2 (1-\mu) \right]r^2 \left(\bar{u},v\right) d\bar{u} < D N^2_\delta,\label{ineq:lphiv}\\
\int_{u_0}^v \left[ \left(\partial_v T \phi\right)^2 + \left(\partial_v \phi\right)^2 + \phi^2 (1-\bar{\mu}) + \left(T\phi\right)^2 (1-\mu)\right]r^2 \left(u,\bar{v}\right) d\bar{v} < D N^2_\delta, \label{ineq:lphiu}
\end{eqnarray}
where $D$ is a constant depending\footnote{In particular, $D$ is independent of $b$. We recall that this follows since the error terms in the energy estimate are all spacetime terms which are $\delta$-small.} only on $M$, $l$, $a$ and where \\$N_\delta:=N_{inv}\left[\bar{r},\bar{\phi}, (u_0,u_0+\delta]\right]$.

Moreover, one has also the pointwise estimates
\begin{eqnarray}
|\phi| &< C_b \ N_\delta \ r^{-\frac{3}{2}-\frac{s}{2} } \label{ineq:lphi2} \, , \\ 
|\phi_u| + |\phi_v| &< C_b \ N_\delta \ r^{-\frac{1}{2}-\frac{s}{2}} \label{ineq:lphi} \, .
\end{eqnarray}
where $C_b>0$ depends only on $b$ (and the fixed parameters $M$, $l$, $a$) and $s$ is the constant defined in \eqref{sdef}.
\end{proposition}
\begin{proof}
The existence and uniqueness of the solution, as well as the energy estimates \eqref{ineq:lphiv}-\eqref{ineq:lphiu}, are direct consequences of Theorem 6.1 of \cite{Holzegelwp}. For the pointwise estimates, we proceed as follows.
First, one obtains easily the pointwise bound 
$$
|r^{3/2}\phi_u | \le D N_\delta^{1/2} \, , 
$$
by integrating the wave equation (viewed as an equation on $r \phi_u$ in the $v$-direction) from the data and applying the energy estimate. To obtain the improved pointwise bound on $\phi_u$, we define the quantity $A(u,v)$,
$$
A=r^n  \frac{r \phi_u} {r_u} + 2 p r^n \phi=  r^n \frac{r \phi_u}{\tilde{r}_u (1-\mu_M )} + 2 p r^n \phi, 
$$ 
where $n$ is a positive real number and $p=\frac{3}{4}-\sqrt{ \frac{9}{16}+ \frac{a}{2}}$. Note that
\begin{eqnarray} \label{eq:p}
1 + \frac{a}{2p}-\frac{1}{2}(2p-1)=0 \, .
\end{eqnarray}
One derives an evolution equation for $A$ using the wave equation for $\phi$:
$$
\partial_v A = A h+ f,
$$
where
$$h= \frac{r_v}{r}\left( n + 2p -1 \right)- \frac{r_v}{1-\mu_M}\left[ 2\frac{\varpi}{r^2}+ 2\frac{r}{l^2}\right]
$$
and
\begin{eqnarray}  \label{eq:ferror}
f&=& r^{n+1} \frac{\phi_u}{\tir_u^2} \frac{\tir_{uv}}{ 1-\mu_M}+ r^n \kappa(2p-1) \mathcal{T}(\phi)  \nonumber \\
&&\hbox{}+ \phi r^n 2p r_v \left( \frac{1}{1-\mu_M}\left( \frac{2 \varpi}{r^2}+ \frac{2r}{l^2}\right) -(2p-1) \frac{1}{r} +  \frac{a r}{ l^2 p (1-\mu )} \right),
\end{eqnarray}
with $\mathcal{T}(\phi)= \frac{1}{\kappa}\partial_v \phi + \frac{1}{\gamma} \partial_u \phi$ and $\gamma=\frac{-r_u}{1-\mu}$, $\kappa=\frac{r_v}{1-\mu}$. (Note that, interestingly the $\mathcal{T}\left(\phi\right)$-term drops out in the conformally coupled case, $a=-1$, but not in general). Without loss of generality, we can assume that $(u,v)$ are asymptotically-AdS coordinates (cf.~Remark \ref{coc}). In these coordinates $|\kappa - \frac{1}{2}| \leq C_b \tilde{r}^3$ and $|\gamma - \frac{1}{2}| \leq C_b \tilde{r}^3$ which implies that $\mathcal{T}(\phi)$ satisfies the same pointwise bound as $T(\phi)$ in the triangle. A pointwise estimate on $T(\phi)$ in turn follows from the $H^1$ bound (\ref{ineq:lphiv}). In summary, we have
$$|\mathcal{T}(\phi)r^{3/2}| \le C_{b} N_\delta^{1/2}.$$
The term containing $\tir_{uv}$ in (\ref{eq:ferror}) decays like $r^{n-7/2}$, using the pointwise bound on $\phi_u$ obtained from the energy estimates and the decay of $\tir_{uv}$.
Note that the term in the bracket on the last line of \eqref{eq:ferror} decays in fact as $\frac{1}{r^2}$ as the leading terms in $\frac{1}{r}$ cancel each other in view of equation \eqref{eq:p}. 
Since $h= \frac{\kappa}{r} \left(  \frac{n-3 +2p}{l^2}r^2 + h' \right)$
where $h'$ is uniformly bounded by a constant depending on $b$, we choose $n=\min ( 3 -2p, 2)$. 
For this choice, one has in particular that $\int_{u_0}^{v} | f |(u,v') dv' < C_{b}N_\delta^{1/2}$. Indeed, one has for instance:
$$
\int_{u_0}^v |\phi| r^n dv \le C_{b}N^{1/2} \int_{u_0}^v r^{\min ( 3 -2p, 2)} r^{-3/2} \frac{r_v}{r^2}dv \le C_{b}  N_\delta^{1/2}.
$$
Hence, it follows that the quantity $A$ is uniformly bounded:
$$
|A(u,v) | \le C_{b}  N_\delta^{1/2}.
$$
Using this estimate, we may now re-estimate $\phi$ from infinity by integrating $r^{2p} \phi$. Since $2p < 3/2$, $r^{2p} \phi$ vanishes at the boundary we have
\begin{eqnarray}
|r^{2p} \phi| &\le& \int^u_{v} |\partial_u \left( r^{2p} \phi \right)|(u',v) du' \nonumber \\
&\le& C_{b}N_\delta^{1/2} \int^u_{v} r^{\max( 4p -4, 2p -3)}(-r_u) (u',v)du' \nonumber \\
&\le& C_{b} N_\delta^{1/2} r^{\max( 4p -3, 2p -2)},\nonumber
\end{eqnarray}
with $C_b$ depending also on $a$, which is the estimate (\ref{ineq:lphi2}). The improved estimate on $\phi_u$, (\ref{ineq:lphi}), then follows from the one already derived for $A$. Finally, one obtains the improved estimate on $\phi_v$ from the improved estimates on $\phi_u$ and $T(\phi)$.
\end{proof}
\subsection{The contraction map in renormalized variables}
Consider a renormalized initial data set $\left(\bar{\tir},\bar{\tio},\bar{\phi}, \bar{\varpi}\right)$ as introduced in Definition \ref{lem:idrnv}. Note that by reducing the $\delta$ in $\Delta_{\delta,u_0}$ one can achieve that the estimates of Lemma \ref{lem:btd} hold on the data for an arbitrary $\delta^\prime>0$.
We now define the following map $\Phi$ on the domain $B_{\mathcal{C},b}$: $\Phi(\tir,\tio,\varpi,\phi)=(\hat{\tir},\widehat{\tio},\widehat{\varpi},\widehat{\phi})$ where:

\begin{itemize}
\item{$\wtr$ and $\widehat{\tio}$ and are defined as:
\begin{eqnarray}
\hat{\tilde{r}} &=& 0 + \bar{\tir}(u)-\bar{\tir}(v)\nonumber \\ 
&&\hbox{}+\int^u_{v} du^\prime \int^v_{u_0} -\tio^2\left[ \frac{2\pi r a \phi^2}{l^2}+\frac{M-\varpi}{2 \left(1-\mu_M\right) r^2}\left(1+ \frac{3r^2}{l^2}\right)\right] dv', \nonumber \\
\log \left(\widehat{\tio}^2\right) &=& \log \left(-\frac{4 \bar{\tilde{r}}_u \overline{\tilde{r}_v} \left(1-\overline{\mu_M}\right)}{1-\bar{\mu}} \right) \left(u\right)\nonumber \\
&&\hbox{} - \log \left(-\frac{4 \bar{\tilde{r}}_u \overline{\tilde{r}_v} \left(1-\overline{\mu_M}\right)}{1-\bar{\mu}} \right) \left(v\right)   
+ \log \left(-4\widehat{\tilde{r}}_u \widehat{\tilde{r}}_v\right) \left(v,v\right)
\nonumber \\ &&\hbox{}+\int^u_{v} du^\prime \int^v_{u_0} \Bigg[-8 \pi \partial_u \phi \partial_v \phi  + \frac{\tilde{\Omega}^2}{r^3} \left(\varpi-M\right) \nonumber\\
&&\hbox{}- \left(\frac{2M}{r^2} + \frac{2r}{l^2}\right) \tio^2\left[ \frac{2\pi r a \phi^2}{l^2}+\frac{M-\varpi}{2 \left(1-\mu_M\right) r^2}\left(1+ \frac{3r^2}{l^2}\right)\right]\Bigg] dv'. \nonumber
\end{eqnarray}
\item $\hat{\varpi}$ is defined for each $v$ as the unique solution of
\begin{align} \label{varpidefe}
\widehat{\varpi}_u = -8\pi r^2 \frac{-\tir_{v}}{\tio^2} \left( \partial_u \phi \right)^2 + \frac{4\pi a}{l^2}  r^2 \tir_u(1-\mu_M) \phi^2,  \textrm{ \ \ \ \ \  $\varpi\left(v,v\right)=M$,}
\end{align}
where $r$ is defined as in Lemma \ref{lem:asarv}.
} 
\item Finally, $\hat{\phi}$ is the unique solution in $C^1(H^{1}_{AdS})$ of the initial boundary value problem $\Box_g \hat{\phi} +\frac{\alpha}{l^2} \hat{\phi} = 0$, where $g$ is the asymptotically AdS spacetime associated with $\left(\tir,\tio,\varpi,\phi\right)$ via Lemma \ref{coads}.
\end{itemize}

As an immediate consequence of the definition of $\Phi$, we note
\begin{lemma}
The hatted functions restrict to the prescribed initial data on
$v=u_0$. The functions $\hat{\tir}$, $\hat{\tir}_u$, $\hat{\tir}_v$, $\hat{\varpi}$ and $\widehat{\tio}$ extend continuously to the boundary $u=v$. On $u=v$, we have:
\begin{align}
\hat{\tilde{r}}(v,v)= 0 \, \ \ , \quad \hat{\tilde{r}}_u + \hat{\tilde{r}}_v =0 \, \ \ , \quad
\hat{\varpi}(v,v)=0 \, \ \ ,  \quad \widehat{\tio}(v,v)=\left(  -4\hat{\tilde{r}}_u \hat{\tilde{r}}_v\right)^{1/2}(v,v)\, . \nonumber
\end{align}
\end{lemma}
We shall prove that if the size of the triangular domain $\delt$ and hence $\delta > 0$ is chosen small enough, depending only on $\mathcal{A}\left[\mathcal{N}_1\right]$ and $N_{inv}[\bar{r},\bar{\phi}, \mathcal{N}_1]$, then $\Phi$ is a contracting map. 
\subsection{Uniform estimates} \label{se:ue}
We now fix the size of the ball to be 
\begin{align} \label{eq:choiceofb}
b= 200 D^{1/2}N_{inv}\left[\bar{r},\bar{\phi}, \mathcal{N}_1\right]
\end{align}
where $D$ is the uniform constant of Proposition \ref{lineartheory}.

We shall establish in this section:
\begin{lemma}
If $\delta$ is small enough, depending only on $b$, $M$ and $l$, the map $\Phi$ maps from $B_{\mathcal{C},b}$ into itself.
\end{lemma}

\begin{proof}
In view of Lemma \ref{lem:asarv}, the estimates \eqref{es:tirrinver} hold. We shall use these bounds without further reference in the remainder of this proof. Moreover, we shall denote by $C_b >0$, a constant depending only on $b$ (and $M$, $l$) which may change from line to line. 
\subsubsection*{Uniform bounds on $\tir$:}

\begin{eqnarray}
\hat{\tir}\left(\frac{u-v}{2}\right)^{-1} &\le& 2\sup \bar{\tir}_u\nonumber\\
&&\hbox{}+2\sup_{u} \int^v_{u_0}-\tio^2\bigg[ \frac{2\pi r a \phi^2}{l^2}
+\frac{M-\varpi}{2(1-\mu_M)r^2}\left(1+ \frac{3r^2}{l^2}\right)\bigg]dv', \nonumber \\
&\le&  e^{\delta'}+ C_b \sup_u \int_{u_0}^{v}\bigg[r^{-3} r^4 \phi^2+ r^{-2} \frac{M-\varpi}{2(1-\mu_M)}\left(1+ \frac{3r^2}{l^2}\right) \bigg]  dv',  \nonumber \\
&\le& e^{\delta'}+C_b \delta^{3}, \nonumber
\end{eqnarray}
where we have used the $\rho$-weighted integral bounds to control the matter term. Since we can estimate $\hat{\tir}\left(\frac{u-v}{2}\right)^{-1}$ from below similarly, we may ensure, by choosing $\delta$ and $\delta'$ small enough that
$$|\log \hat{\tir}\left(\frac{(u-v)}{2}\right)^{-1}| \leq  \frac{b}{100} \, .$$
Similarly, we have
\begin{eqnarray}
\hat{\tir}_v (u,v)&=& -\bar{\tir}_u(v) + \int^u_{v} -\tio^2\bigg[ \frac{2\pi r a \phi^2}{l^2}
+\frac{M-\varpi}{2(1-\mu_M)r^2}\left(1+ \frac{3r^2}{l^2}\right)\bigg](u',v)du' \nonumber \\
&&\hbox{}-\int^v_{u_0}\tio^2\bigg[ \frac{2\pi r a \phi^2}{l^2}+\frac{M-\varpi}{2(1-\mu_M)r^2}\left(1+ \frac{3r^2}{l^2}\right)\bigg](v,v')dv', \nonumber
\end{eqnarray}
and estimating the integral terms as above, we may ensure that
$$\left|\log \left(-2 \hat{\tir}_v \right)\right|  \leq  \frac{b}{100}.$$
Obviously, a similar estimate holds for $\hat{\tir}_u$. 

Using the pointwise bound on $\phi_u$ (cf.~(\ref{pwpu})), we derive pointwise bounds on $\phi$ by integration along the $v$=const lines:
$$
| {\phi}(u,v) | \le C_{b} \rho^{3/2 +s/4}.
$$
Using this pointwise estimate, we can then easily bound $\rho^{-1} T(\hat{\tir})=\rho^{-1} \hat{\tir}_u+\rho^{-1} \hat{\tir}_v$:
\begin{align}
|\rho^{-1} T(\hat{\tir})| \le \sup\left| \rho^{-1}(\bar{\tir}_u\left(u\right)-\bar{\tir}_u\left(v\right)) \right| +C_b \rho^{-1} \int^u_{v} \rho^2(u',v) du' 
\nonumber \\ + \rho^{-1} \int^v_{u_0}(v-u) C_b dv' 
 \le 2\delta'+C_b \delta \le \frac{b}{100}. \nonumber
\end{align}
We can also estimate the second derivative $\hat{\tilde{r}}_{uv}$: 
\begin{eqnarray}
|\hat{\tir}_{uv}|&=&\bigg|-\tio^2\left[ \frac{2\pi r a \phi^2}{l^2}+\frac{M-\varpi}{2(1-\mu_M)r^2}\left( 1+ \frac{3r^2}{l^2}\right)\right] \bigg| \nonumber \\
&\le& C_{b} \rho^{2 +s/2}+C_{b}\rho^{2+s/4}, \nonumber
\end{eqnarray}
where the first term from the right-hand side is obtained from the pointwise estimate on $\phi$ and the second from the bound on $|\varpi-M|$.
Hence, we have:
$$
|\rho^{-2}\hat{\tir}_{uv}| \le C_{b} \delta^{s/4} \le \frac{b}{100}.
$$
$T(\hat{\tir}_{uv})$ may be estimated similarly, using the bound on $\rho^{s/8}T(\varpi)$. To estimate $\hat{\tilde{r}}_{vv}$, we first compute:
\begin{eqnarray}
\hat{\tir}_{vv}&=&-\bar{\tir}_{uu}(v) + \int^u_{v} - 2\tio \tio_{v} \bigg[ \frac{2\pi r a \phi^2}{l^2}
+\frac{M-\varpi}{2(1-\mu_M)r^2}\left(1+ \frac{3r^2}{l^2}\right)\bigg](u',v)du' \nonumber \\
&&\hbox{}+\int^u_{v} - \tio^2 \bigg[ \frac{2\pi r_v a \phi^2}{l^2}+\frac{2\pi r a  2\phi \phi_v}{l^2}
-\frac{\varpi_v}{2(1-\mu_M)r^2}\left(1+ \frac{3r^2}{l^2}\right) \nonumber \\
&&\hbox{}+(M-\varpi)\left( \frac{1}{2(1-\mu_M)r^2}\left(1+ \frac{3r^2}{l^2}\right) \right)_v\bigg](u',v)du' \nonumber \\
&&\hbox{}+ \int^v_{u_0} - 2\tio \tio_{v} \bigg[ \frac{2\pi r a \phi^2}{l^2}
+\frac{M-\varpi}{2(1-\mu_M)r^2}\left(1+ \frac{3r^2}{l^2}\right)\bigg](v,v')dv' \nonumber \\
&&\hbox{}+\int^v_{u_0} - \tio^2 \bigg[ \frac{2\pi r_v a \phi^2}{l^2}+\frac{2\pi r a  2\phi \phi_v}{l^2}
-\frac{\varpi_v}{2(1-\mu_M)r^2}\left(1+ \frac{3r^2}{l^2}\right) \nonumber \\
&&\hbox{}+(M-\varpi)\left( \frac{1}{2(1-\mu_M)r^2}\left(1+ \frac{3r^2}{l^2}\right) \right)_v\bigg](v,v')dv' \nonumber, 
\end{eqnarray}
from which one obtains easily that
$$
|\hat{\tir}_{vv}| \le \sup \left| \bar{\tir}_{uu} \right|+C_{b} \delta+ C_{b} \int^u_{v} \left(r \phi \phi_v \right)(u',v) du'+C_{b} \int^v_{u_0} \left(r \phi \phi_v \right)(v,v') dv'.
$$
The $v$-integral is easily estimated using Cauchy-Schwarz and the weighted-$H^1$ bound on $\phi$. For the $u$-integral, we first use that $\phi_v=T(\phi)-\phi_u$, and then applies Cauchy-Schwarz and the weighted-$H^1$ bounds for $\phi$ and $T(\phi)$. Hence, it follows that:
\begin{eqnarray} \label{es:tirvv}
|\hat{\tir}_{vv}|\le \delta'+C_b \delta \le \frac{b}{100}.
\end{eqnarray}
A similar estimate naturally holds for $\hat{\tilde{r}}_{uu}$. 
In summary, this shows that $$d_{\tilde{r}} \left(\hat{\tir}, \frac{u-v}{2} \right)  \leq \frac{8b}{100} \, . $$

\subsubsection*{Uniform bounds on $\tio$:}
We have:
\begin{flalign}
&\left| \log  \frac{\widehat{\tio}^2}{ -4\hat{\tilde{r}}_u \hat{\tilde{r}}_v \left(v,v\right)} \right| \le \left| \log \frac{\bar{\tir}_u \overline{\tir_v}(u)}{\bar{\tir}_u \overline{\tir_v}(v)} \right| + \left| \log \frac{(1-\overline{\mu_M})(u)(1-\bar{\mu})(v)}{(1-\overline{\mu_M})(v)(1-\overline{\mu})(u)} \right|\nonumber \\
&\hbox{}\qquad \qquad+ |u-v|  \sup_{u'}  \int^v_{u_0}\Bigg[ - 8 \pi \partial_u \phi \partial_v \phi + \frac{\tio^2}{r^3}\left(\varpi-M \right) \nonumber\\
&\hbox{}\qquad\qquad-2 \left( \frac{2M}{r^2}+\frac{2r}{l^2}\right) \tio^2 \left[ \frac{2\pi r a \phi^2}{l^2}+\frac{M-\varpi}{2(1-\mu_M)r^2} (1+ \frac{3r^2}{l^2})\right] \Bigg] dv'. \nonumber
\end{flalign}
Hence, using Cauchy-Schwarz, the estimates on $\hat{\tir}_u$ and $\hat{\tir}_v$ derived above, the pointwise estimates on $\phi_u$ and the $H^1_{AdS}$ bound to control the $\phi_u \phi_v$ term in the integrals, we obtain that:


$$
\left| \log \widehat{\tio}^2 \right| \le  C \delta' + C_b \delta \le \frac{b}{100},
$$
choosing $\delta'$ and $\delta$ sufficiently small, depending only on $b$.
The derivatives of $\tio$ are estimated similarly. For instance,
\begin{eqnarray}
2\left| \frac{\widehat{\tio}_v}{\tio} \right| &\le& \left| \frac{(\overline{\tir_v})_u(v)}{r_{v}(v)} \right|+\left| \frac{\bar{\tir}_{uu}(v)}{r_{u}(v)} \right| \nonumber\\
&&\hbox{}+ \left| \frac{(1-\bar{\mu})_u(v)}{(1-\bar{\mu})(v)} \right|+ \left| \frac{(1-\overline{\mu_M})_u(v)}{(1-\overline{\mu_M})(v)} \right| \nonumber\\
&&\hbox{}+ \left| \frac{\hat{\tir}_{uv}}{\hat{\tir}_u}\right|(v,v)+\left| \frac{\hat{\tir}_{uv}}{\hat{\tir}_v}\right|(v,v)+\left| \frac{\hat{\tir}_{uu}}{\hat{\tir}_u}\right|(v,v)+\left| \frac{\hat{\tir}_{vv}}{\hat{\tir}_v}\right|(v,v) \nonumber \\
&&\hbox{}+\Bigg|\int^v_{u_0}\Bigg[ - 8 \pi \partial_u \phi \partial_v \phi + \frac{\tio^2}{r^3}\left(\varpi-M \right) \nonumber\\
&&\hbox{}-2 \left( \frac{2M}{r^2}+\frac{2r}{l^2}\right) \tio^2 \left[ \frac{2\pi r a \phi^2}{l^2}+\frac{M-\varpi}{2(1-\mu_M)r^2} (1+ \frac{3r^2}{l^2})\right] \Bigg](v,v') dv'\Bigg|\nonumber\\
&&\hbox{}+\Bigg|\int^u_{v}\Bigg[ - 8 \pi \partial_u \phi \partial_v \phi + \frac{\tio^2}{r^3}\left(\varpi-M \right) \nonumber\\
&&\hbox{}-2 \left( \frac{2M}{r^2}+\frac{2r}{l^2}\right) \tio^2 \left[ \frac{2\pi r a \phi^2}{l^2}+\frac{M-\varpi}{2(1-\mu_M)r^2} (1+ \frac{3r^2}{l^2})\right] \Bigg](u',v) du' \Bigg|,\nonumber
%
\end{eqnarray}
The terms in the first two lines on the right-hand side are controlled by $C\delta'$ using the assumptions on the initial data. We use the estimates on $\hat{\tir}$ derived earlier to control the terms in the third line by $C_b \delta+ C\delta'$. To control the matter terms in the integral, we use the pointwise estimates on $\phi_u$ and the $H^1_{AdS}$-bound on $\phi$ and $T(\phi)$, the latter being used in particular to control the $u$ integrals of $\phi_v=T(\phi)-\phi_u$. Hence one obtains

$$
\left| \frac{\widehat{\tio}_v}{\tio} \right| \le C_b \left(  \delta'+\delta \right) \le \frac{b}{100}.
$$

Since a similar estimate holds for $\partial_u \log \widehat{\tio}$ we may ensure, by choosing $\delta$ and $\delta'$ small enough that
$$
\widehat{\tio} \in B_{C^1_{\tio},\frac{3b}{100}}.
$$

\subsubsection*{Uniform estimates for $\phi$} 
By Proposition \ref{lineartheory}, we immediately obtain
\begin{align}
\| \widehat{\phi} \|_{C^1\left(H^1_{AdS} \right)} \leq D^{1/2} N\left[\bar{r},\bar{\phi}, \left(u_0,u_0+\delta\right]\right] < \frac{b}{100}, \,
\end{align}
in view of the definition of $b$, (\ref{eq:choiceofb}).
Moreover, we have a pointwise estimate on $\widehat{\phi}_u$. Applying (\ref{ineq:lphi}) we derive
$$
| \rho^{-\frac{1}{2}-s/4}\widehat{\phi}_u | \le C_b \rho^{s/4} \le \frac{b}{100}.
$$




\subsubsection*{Uniform estimates for $\varpi$}
Using the pointwise bounds on ${\phi}_u$ and ${\phi}$, we estimate
\begin{eqnarray}
\rho^{-\frac{s}{4}}|\widehat{\varpi} -M|
&=& \rho^{-\frac{s}{4}} \int^u_{v}\left[-8\pi r^2 \frac{-\tir_{v}}{\tio^2} \left( \partial_u \phi \right)^2 + \frac{4\pi a}{l^2} r^2 \tir_u(1-\mu_M) \phi^2  \right] \left(u^\prime,v\right) du^\prime \nonumber \\
&\leq& C_b \rho^{\frac{s}{4}} \leq \frac{b}{100} \nonumber
\end{eqnarray}
A pointwise estimate for $\partial_u \widehat{\varpi}$ is immediate from (\ref{varpidefe}) itself because we have the improved pointwise estimates on $\phi$ and $\phi_u$:
\begin{align}
| \rho \partial_u \widehat{\varpi}| \le C_{b} \delta'.
\end{align}
To estimate the $v$-derivative, we first commute (\ref{varpidefe}) with $T$ to obtain:
$$
T(\widehat{\varpi})_u = (T\widehat{\varpi})_u=  T\left( -8 \pi r^2 \frac{\tir_v}{\tio^2} \phi_u^2+ \frac{4 \pi a}{l^2}r^2 \tir_u (1-\mu_M) \phi^2 \right).
$$

In view of the fact that the mass is constant on the boundary, we have $T(\widehat{\varpi})=0$ on the boundary. Hence we may integrate the commuted equation with trivial boundary data upon which we obtain, as above
\begin{align}
|T\widehat{\varpi} |
=  \int^u_{v}\left[ -8\pi r^2 \frac{-\tir_{v}}{\tio^2} 2 \left( \partial_u \phi \right) \left( \partial_u T\phi \right)+ ... \right] \left(u^\prime,v\right) du^\prime \, . \nonumber
\end{align}
Once again, we apply Cauchy-Schwarz and then use the improved pointwise estimates on $\phi$ and $\phi_u$, to obtain:
$$
\rho^{-s/8}|T\widehat{\varpi} | \le C_b \rho^{s/8} \le  \frac{b}{100}.
$$
The inequality $| \rho \partial_v \widehat{\varpi}| \leq \frac{b}{100}$ then follows, since $\partial_v = T - \partial_u$. Collecting our estimates we have indeed shown
\begin{align}
d\left(\left(\hat{\tilde{r}}, \widehat{\varpi}, \hat{\tilde{\Omega}}, \hat{\phi}\right) - \left(\frac{u-v}{2}, 1, M, 0\right)\right) < b \nonumber
\end{align}
\end{proof}
\subsection{Estimating differences}
To establish the contraction property for $\Phi$, we will need to estimate differences. As it turns out, we will only be able to show the contraction property in a weaker norm (involving a derivative less, roughly speaking) and retrieve the full regularity of the fixed point a-posteriori from standard arguments.
\begin{definition}
On $\mathcal{B}_{\mathcal{C},b}$, we define the weaker distance
\begin{eqnarray}
&&d_{w} \left(\left(\tilde{r}_1, \tilde{\Omega}_1, \varpi_1, \phi_1\right), \left(\tilde{r}_2, \tilde{\Omega}_2, \varpi_2, \phi_2\right)\right)= || \log \frac{\tir_1}{\tir_2} ||_{C^0(u,v)}
\nonumber \\
&&\hbox{}\qquad + || \rho^{-1} \left[ T\left(\tilde{r}_1\right) - T\left(\tilde{r}_2\right)\right]||_{C^0(u,v)} + || \log [(\tir_1)_u] -  \log [(\tir_2)_u] ||_{C^0} \nonumber \\ 
&&\hbox{}\qquad+|| \log [(-\tir_1)_v] -  \log [-(\tir_2)_v] ||_{C^0}   + | \varpi - M|_{C^0(u,v)} + \|\phi\|_{C^0\left(H^{1}_{AdS}\right)} \nonumber \\  
&&\hbox{}\qquad
+\Big\| \log \left(\tio_1^2\right) -  \log \left(\tio_2^2\right)  \Big\|_{C^0(u,v)}
\end{eqnarray}
and we denote by $\overline{\mathcal{B}}_{w,b}$, the closure of $\mathcal{B}_{\mathcal{C},b}$ with respect to $d_{w}$.
\end{definition}
\vspace{.2cm}
%
\begin{lemma} \label{lem:fp}
If $\delta$ is sufficiently small, depending only on $\mathcal{A}[\mathcal{N}_1]$, the invariant norm $N_{inv}[\bar{r},\bar{\phi},\mathcal{N}_1]$ and the size of the ball $b$, the map $\Phi$ is a contraction with respect to the distance $d_{w}$ and hence has a unique fixed point in $\overline{\mathcal{B}}_{w,b}$.
\end{lemma}

\begin{proof}

Let $q=(\tir,\tio,\phi,\varpi)$ and $q'=(\tir',\tio',\varpi',\phi')$ be in $B_{\mathcal{C},b}$. Let $\hat{q}=\Phi(q)$ and $\hat{q'}=\Phi(q')$. We want to prove that: 
$$d_{w} ( \hat{q},\hat{q'} ) \le \delta \cdot C_b \cdot d_{w} (q,q'),$$
where $C_{b}$ is a constant depending only on $b$.
The estimates for the variables $\tir,\Omega,\varpi$ will carry through for differences as before without additional difficulty. Hence we focus on the estimate for $\psi=\widehat{\phi}-\widehat{\phi'}$.

Denote by $g$, $g^\prime$ the asymptotically AdS-metric associated to $q$ and $q'$ by Lemma \ref{coads}. Let us also write $\Omega^2_{aux} = -4 \frac{\tilde{r}_u \tilde{r}_v}{1-\mu} \left(1-\mu_M\right)^2$. The difference $\psi$ satisfies
\begin{eqnarray} \label{eq:psi}
\Box_g \psi - \frac{2a}{l^2} \psi= \frac{4}{\left(\tilde{\Omega}_{aux}\right)^2}\ \Ecal 
\end{eqnarray}
where 
\begin{align}
\Ecal = \widehat{\phi'}_v \left(\frac{{r}^\prime_u}{{r}^\prime} - \frac{{r}_u}{{r}} \right) 
+  \widehat{\phi'}_u \left(\frac{{r}^\prime_v}{{r}^\prime} - \frac{{r}_v}{{r}} \right) 
+ \frac{a}{2l^2} \widehat{\phi'}\left(\left(\Omega_{aux}^\prime\right)^2 - \left(\Omega_{aux}\right)^2\right)  \nonumber \\
= \widehat{\phi'}_v \frac{1}{{r}^\prime} \left( {r}_u^\prime - {r}_u  \right) 
- \widehat{\phi'}_v {r}_u \left( \frac{1}{r} -\frac{1}{{r}^\prime} \right) 
+ \widehat{\phi'}_u \frac{1}{{r}^\prime} \left( {r}_v^\prime - {r}_u  \right) 
- \widehat{\phi'}_u {r}_v \left( \frac{1}{r} -\frac{1}{{r}^\prime} \right) 
\nonumber \\
- \frac{2a}{l^2} \widehat{\phi'} \left[\tilde{r}_u^\prime \tilde{r}_v^\prime \left(\frac{\left(1-{\mu}_M^\prime\right)^2}{1-\mu^\prime} - \frac{\left(1-{\mu}_M\right)^2}{1-\mu}\right) + \frac{\left(1-\mu_M\right)^2}{1-\mu} \Big[\tilde{r}_u^\prime \left(\tilde{r}_v^\prime -\tilde{r}_v \right)+ \tilde{r}_v  \left(\tilde{r}_u^\prime - \tilde{r}_u \right) \Big] \right] \nonumber
\end{align}
which we will decompose as $\mathcal{E} = \mathcal{E}_1 + \mathcal{E}_2 + \mathcal{E}_3$.
\begin{align}
\mathcal{E}_1 &=  \widehat{\phi'}_v \frac{1}{{r}^\prime} \left( {r}_u^\prime - {r}_u  \right)  +  \widehat{\phi'}_u \frac{1}{{r}^\prime} \left( {r}_v^\prime - {r}_u  \right)  \nonumber \\
& \ \ \ - \frac{2a}{l^2} \widehat{\phi'}  \frac{\left(1-\mu_M\right)^2}{1-\mu} \Big[\tilde{r}_u^\prime \left(\tilde{r}_v^\prime -\tilde{r}_v \right)+ \tilde{r}_v  \left(\tilde{r}_u^\prime - \tilde{r}_u \right) \Big] 
\nonumber 
\end{align}
\begin{align}
\mathcal{E}_2 &=  - \widehat{\phi'}_v {r}_u \left( \frac{1}{r} -\frac{1}{{r}^\prime} \right) - \widehat{\phi'}_u {r}_v \left( \frac{1}{r} -\frac{1}{{r}^\prime} \right) \nonumber \\
& \ \ \ -2 \frac{a}{l^2} \widehat{\phi'} \frac{\tilde{r}_u^\prime \tilde{r}_v^\prime}{\left(1-\mu^\prime\right)\left(1-\mu\right)} \Bigg[ \frac{r^2 \left(r^\prime\right)^2}{l^6} \left(r+r^\prime\right) + P_3\left(r,r'\right) \Bigg] \left(r-r^\prime\right)
\end{align}
\begin{align}
\mathcal{E}_3 =  -2 \frac{a}{l^2} \widehat{\phi'}\frac{\tilde{r}_u^\prime \tilde{r}_v^\prime}{\left(1-\mu^\prime\right)\left(1-\mu\right)}  \left[\frac{2 r^4}{l^4 r^\prime} + P_2(r,r') \right] \left(\varpi - \varpi^\prime\right),
\end{align}
where $P_i(r,r')$ satisfies $|\frac{P_i(r,r')}{r^i}|<C_b$.
The energy estimate applied to $\psi$ satisfying (\ref{eq:psi}) yields:
\begin{eqnarray} 
\int^u_{v} \left( r^2 \psi^2_u +r^4 \psi^2\right)(u',v) du'+\int^u_0 \left( r^2 \psi^2_v +r^4 \psi^2 \right)(u,v') dv' \le \nonumber \\
C \int^u_{u_0} \left( r^2 \psi^2_u +r^4 \psi^2\right)(u',u_0) du'+ C \int_{u_0}^v \int_{v'}^u T\left(\psi\right) \mathcal{E} r^2 du' dv', \label{es:lipspi} 
\end{eqnarray}
where $C>0$ is a constant depending only on $a$, $l$, $M$ and $b$. It follows that we only need to estimate the spacetime term in the above.

We first note that for $\mathcal{E}_3$:
\begin{align}
\left|\int_{u_0}^v \int_{v'}^u  T\left(\psi\right) \mathcal{E}_3 r^2 du' dv'\right| \leq \delta \sqrt{\sup_u \int^{v}_{u_0} r^2 \left(\partial_v\psi\right)^2 dv^\prime}\sqrt{\sup_{u} \int^{v}_{u_0} \left( \mathcal{E}_3\right)^2 r^2 dv^\prime} \nonumber \\
+ \delta \sqrt{\sup_{u_0<v^\prime<v} \int_{v^\prime}^u r^2 \left(\partial_u\psi\right)^2 \left(u, v^\prime \right) du^\prime}\sqrt{ \sup_{u_0<v^\prime<v} \int_{v^\prime}^u \left( \mathcal{E}_3\right)^2 r^2 \left(u, v^\prime \right) du^\prime} \nonumber \\
\leq C_{b}\delta d_{w} \left(q,q^\prime\right) \sup_{\Delta_{\delta,u_0}} |\varpi - \varpi^\prime| 
\leq \delta \cdot C_{b} \cdot d_{w}^2 \left(q,q^\prime\right), \nonumber
\end{align}
where $C_{b}$ is a constant depending only on $b$.

The terms $\mathcal{E}_1$ and $\mathcal{E}_2$ are more difficult to estimate. We note
\begin{lemma}
We have the estimates
\begin{align} \label{tyu}
C_b^{-1} | \tilde{r} - \tilde{r}^\prime|  \leq \Big|\frac{1}{r} - \frac{1}{r^\prime}\Big| \leq C_b | \tilde{r} - \tilde{r}^\prime| 
\end{align}
\begin{align}
| r_u - r_u^\prime | \leq C_b \cdot r^2 \cdot | \tilde{r}_u - \tilde{r}^\prime_u |  + C_b  \cdot \delta \cdot r \cdot  | \tilde{r} - \tilde{r}^\prime|  \\
| r_v - r_v^\prime | \leq C_b \cdot r^2 \cdot | \tilde{r}_v - \tilde{r}^\prime_v |  + C_b  \cdot \delta \cdot r \cdot  | \tilde{r} - \tilde{r}^\prime|
\end{align}
\end{lemma}
\begin{proof}
Write
\begin{align}
\left( \tilde{r} - \tilde{r}^\prime \right)_u = \left( \frac{1}{r} - \frac{1}{r^\prime}\right)_u \left\{-\frac{r^2}{1-\mu_M} \right\} + \nonumber \\
 \left\{ \frac{r_u^\prime}{\left(r^\prime\right)^2} \frac{1}{\left(1-\mu_M\right)\left(1-\mu^\prime_M\right)}\left[\left(r+r^\prime\right) r r^\prime - 2M \left(\left(r^\prime\right)^2 + r r^\prime + r^2\right)  \right] \right\}\left( \frac{1}{r} - \frac{1}{r^\prime} \right) \nonumber
\end{align}
and note that the first curly bracket is uniformly bounded, while the second is $\delta$-small (it decays like $\frac{1}{r}$). Estimate (\ref{tyu}) is then easily derived. The remaining estimates follow after revisiting the identity above and using the bound (\ref{tyu}).
\end{proof}
These estimates will be used freely for the differences appearing in the $\mathcal{E}_i$. With the help of the Lemma one observes (simply by counting powers) that the $r$-weights in the spacetime terms arising from $\mathcal{E}_1$ and $\mathcal{E}_2$ are too strong to close the estimate by naively estimating $T\psi$ via the $H^1_{AdS}$-norm of $\psi$ (which is what we are estimating on the left hand side at this stage). This could be resolved using the $H^1_{AdS}$-norm for $T\psi$, which however, is not available at this point. The resolution is to integrate by parts the $T$ derivative, since an improved estimate is available for $\psi$ itself, while the $T$-derivative falling on $\mathcal{E}_{1,2}$ will not lead to a loss in $r$-weight.
Let us consider the spacetime integrals ($i=1,2$)
\begin{align}
I_i \left(u,v\right) = \int_{u_0}^{v} dv' \int_{v^\prime + \eta}^u du^\prime \ T\left(\psi\right) \mathcal{E}_{i} \ r^2 \left(u^\prime, v^\prime\right) \, ,
\end{align}
for $\left(u,v\right) \in \Delta \cap \{ u\geq v+ \eta\}$ and $\eta>0$, the latter region
converging to the original spacetime integrals as $\eta \rightarrow 0$.
We will first move the $T$ derivative onto $\mathcal{E}_{i}$. 
The boundary term coming from $u=v+\eta$ vanishes, since $T$ is tangent to the boundary.
The term on $v=u_0$ (the data) also vanishes, because $\psi = 0$ there. The other boundary terms can be estimated
\begin{eqnarray}
\left| \int_{u_0}^v \psi \mathcal{E}_{i} r^2 \left(u, \bar{v}\right) d\bar{v} \right| \leq \sqrt{ \int_{u_0}^v \psi^2 r^4 dv} \sqrt{\int_{u_0}^v \mathcal{E}_{i}^2 d\bar{v}} 
\leq d_{w}(q,q')\sqrt{\int_{u_0}^v \mathcal{E}_{i}^2 d\bar{v}} \, . \nonumber
\end{eqnarray}
On the other hand, we have:
$$
\sqrt{\int_{u_0}^v \mathcal{E}_{1}^2 d\bar{v}} \le \delta C_{b}\left( \sup |\tir_u -\tir_u' | + \sup |\tir_v -\tir_v' | + \sup |\log \frac{\tir'}{\tir}|\right) \le \delta C_{b} d_{w}(q,q')
$$
and 
$$
\sqrt{\int_{u_0}^v \mathcal{E}_{2}^2 d\bar{v}} \le \delta C_{b} \sup | 1- \frac{\tir'}{\tir} | \le C_{b} \sup |\log \frac{\tir'}{\tir}| \le \delta C_{b} d_{w}(q,q').
$$
Since the $u$-boundary term can be estimated similarly, both boundary terms are controlled by $\delta \cdot C_{b} \cdot d_{w}(q,q')^2$. 

It remains to estimate the resulting spacetime terms from the integration by parts, which are given, for $i=1,2$, by
\begin{align}
 \int_{u_0}^{v} dv' \int_{v^\prime + \eta}^u du^\prime \ \psi \ T \left(\mathcal{E}_{i}  r^2 \right)\left(u^\prime, v^\prime\right) \, .
\end{align}
The point is to observe that taking a $T$-derivative we will not lose a power in $r$. 
We distinguish the following terms arising from  $T \left(\mathcal{E}_{i}\right)$: Terms which do not involve second derivatives of the $\tilde{r}$-difference and terms which do. The terms which contain only one derivative of the $\tir$-difference can be estimated directly. For instance,
\begin{align}
\left| \int \int \left(T\phi^\prime \right)_v \frac{1}{\tilde{r}^\prime} \left(\tilde{r}_u^\prime - \tilde{r}_u\right)\psi \  r^2 du dv \right| \nonumber \\
\leq \sup |\tilde{r}_u^\prime - \tilde{r}_u| 
\cdot \int du \left(  \sqrt{\int_{u_0}^v \left(\partial_v T\phi^\prime\right) r^2 dv} \sqrt{ \int_{u_0}^v \psi^2 r^4 dv}\right) 
\leq  \delta \cdot d_{w}^2 \left(q,q^\prime\right) \, . \nonumber
\end{align}
In particular, all the terms coming from $\Ecal_2$ may be estimated as above. 
It remains to estimate the terms containing two derivatives of the $\tir$-difference, i.e.
\begin{align}
 \int_{u_0}^{v} d\bar{v} \int_{\bar{v} + \eta}^u d\bar{u} \ \psi \  \mathcal{Y}  r^2 \left(u^\prime, v^\prime\right) \, ,
\end{align}
with
\begin{align}
\mathcal{Y} =\widehat{\phi}^\prime_v \frac{1}{r^\prime} \left(T{r}_u^\prime - T{r}_u\right) + \phi^\prime_u \frac{1}{r^\prime} \left(T{r}_v^\prime - T{r}_v\right) \nonumber \\
-2 \frac{a}{l^2} \widehat{\phi^\prime} \frac{\left(1-{\mu}_M\right)^2}{1-\mu} \Big[\tilde{r}_u^\prime \left(T\tilde{r}_v^\prime -T\tilde{r}_v \right)+ \tilde{r}_v  \left(T\tilde{r}_u^\prime - T\tilde{r}_u \right) \Big].
\end{align}
We write $\mathcal{Y}$ schematically as
$$\mathcal{Y}=\mathcal{Z} \partial_u T( \tir-\tir')+\mathcal{W} \partial_v T( \tir-\tir')
$$
and integrate by parts the $u$ and $v$ derivatives. 
The boundary terms are estimated as before, while a typical spacetime term is given by
\begin{align}
\int \int d\bar{u} \ d\bar{v} \partial_u \left( r^2 \psi \frac{\widehat{\phi'}_v}{r^\prime} \right) T( r -r') = \nonumber\\
\int \int d\bar{u} \ d\bar{v}\left( r^2 \psi_u \frac{\widehat{\phi'}_v}{r'} - r^2 \psi \widehat{\phi'}_v \frac{r'_u}{\left(r^\prime\right)^2}+r^2 \psi \frac{\widehat{\phi'}_{uv}}{r'} + \psi \frac{\widehat{\phi'}_{v}}{r'}2r r_u\right) T( \tir -\tir'). \nonumber
\end{align}
We use the wave equation for $\widehat{\phi'}$ to substitute the second order derivative term $\widehat{\phi}'_{uv}$ by lower order terms. All terms are then of the same strength and estimated in the same way. For instance, for the first term
\begin{align}
\int \int d\bar{u} \ d\bar{v} \  r^2 \psi_u \frac{\widehat{\phi'}_v}{r'} T( r - r') \nonumber \\ 
\leq \sup |\rho \cdot T \left(\tilde{r}-\tilde{r}^\prime\right)| \int \int d\bar{u} \ d\bar{v} \  r^2 \psi_u {\widehat{\phi'}_v}  \leq C_b \cdot \delta \cdot d^2_{w} \left(q,q^\prime\right)
\end{align}
after applying Cauchy-Schwarz and using that $T\left(r-r^\prime\right) \leq C_b r^2 | T \left(\tilde{r} - \tilde{r}^\prime\right) |$.
\end{proof}

\subsection{Conclusion of the proof of Proposition \ref{renowp}}
Proposition \ref{renowp} follows directly from Lemma \ref{lem:fp}. Indeed, from Lemma \ref{lem:fp}, we have shown that $\Phi$ has a fixed point $q=\left(\tilde{r},\tilde{\Omega}, \varpi,\phi\right) \in \overline{\mathcal{B}}_{w,b}$ (recall that $b$ was fixed at the beginning of section \ref{se:ue} and $\delta$ was chosen sufficiently small depending only on $b$). Moreover, $q$ is the limit of a sequence in $\mathcal{B}_{\mathcal{C},b}$. Hence, we have in particular that the derivatives of $\tir$ are Lipschitz continuous, $\tio$ is Lipschitz continuous, $\rho^{ -\frac{2+s}{4}} \phi_u$ is uniformly bounded, $\phi$ is in $C^0(H^1_{AdS})$ and $T(\phi) \in L^\infty(H^1_{AdS})$, which implies that $T(\phi)$ is continuous. 
Integrating in $v$ the transport equation on $r\phi_u$ from the initial data, it follows that $\phi_u$ is actually continuous in $\delt$, which implies that $\varpi_u$ is continuous. Since $\phi_v=T(\phi)-\phi_u$, we then have that $\phi_v$ is continuous.
Looking at the equation on $T(\varpi)_u$, it follows that $T(\varpi)$ and hence $\varpi_v$ is continuous. From the wave equation on $\tir$, one now obtains easily that $\tir_{uu}$, $\tir_{vv}$, $\tir_{uv}$ are continuous. From this, it follows that the spacetime $(\delt,g)$ associated to our solution, with $g$ defined as in \eqref{metric:ebcb}, is asymptotically AdS. Hence, from the linear analysis, one infers that $\phi$ is actually $C^1(H_{AdS})$. Any further regularity is then obtained by standard methods. 
%
 Finally, $\delta$ indeed only depends on $\mathcal{A}\left[ \mathcal{N}_1 \right]$ and $N_{inv}[\bar{r},\bar{\phi},\mathcal{N}_1]$. 
Note in this context that if one choses $\tir_u=1-\mu_M$ initially on $\mathcal{N}_1$, then $\mathcal{A}\left[\mathcal{N}_1\right]=0$ and $\delta$ only depends on the value of $N_{inv}\left[\bar{r}, \bar{\phi},\mathcal{N}_1\right]$.

\section{Proof of Theorem \ref{wellposed}} \label{se:mtp}
Theorem \ref{wellposed} is a consequence of Proposition \ref{renowp} and the following 
\begin{proposition} \label{eq}
For a solution $\left(\tilde{r},\tilde{\Omega}, \phi, \varpi\right)$ of the renormalized system as arising from Proposition \ref{renowp}, the tuple $\Big(\Omega^2 = \tilde{\Omega}^2 \left(1-{\mu}_M\right), r, \phi, \varpi\Big)$ is a solution of the original system in the triangle $\Delta_{\delta,u_0}$ and conversely. 
\end{proposition}

\begin{proof}


From Lemma \ref{lem:asarv}, we have:
$$
r_v = \tir_v (1-\mu_M), \quad r_u = \tir_u(1-\mu_M).
$$
Let $\Omega^2=(1-\mu_M)\tio^2$. An easy computation then shows that if the equations of Proposition \ref{renowp} (and hence the equations of Corollary \ref{cor:progconst}) hold, then the original Einstein scalar-field equations \eqref{cons1}-\eqref{laste} hold for the variables $\left(\Omega^2, r, \phi, \varpi\right)$. The converse statement is obtained similarly.
\end{proof}
We can finally conclude the proof of Theorem \ref{wellposed}
\begin{proof}(Theorem \ref{wellposed})
We have a solution to the equations by Proposition \ref{eq}. Moreover,
since the renormalized solution we constructed lives in the ball $B_{\mathcal{C},b}$, we obtain automatically the regularity stated for $\left(\Omega^2, r, \phi, \varpi\right)$ in the first part of Theorem \ref{wellposed}, as well as the statement about the time of existence of the solution (item $1$). Item $2$ of Theorem \ref{wellposed} follows from Lemma \ref{coads}. Finally, to establish item $3$ we go to the coordinate system in which $\tilde{r}_u = -\tilde{r}_v = \frac{1}{2}$ on $\mathcal{I}$ (as we did in Lemma \ref{coads}). This yields all the properties of an asymptotically AdS initital data set uniformly on any fixed $v=const$ rays, except for the pointwise estimate on $\phi_{uu}$, which was assumed for the construction of the data. To retrieve this for the solution, compute first, setting $A_0:= \frac{1}{r_u} \partial_u \left( \frac{\phi_u}{r_u}\right)$:
\begin{eqnarray}
\partial_v A_0 &=& -4 \kappa \left( \frac{\varpi }{r^2} + \frac{r}{l^2}- \frac{4 \pi a \phi^2}{l^2} + \frac{1-\mu }{4r} \right) A_0 \nonumber \\
&&\hbox{}+ \frac{2}{r^2} \phi_v + \frac{2a \kappa}{l^2} \frac{\phi_u}{r_u} \nonumber \\
&&\hbox{}- 8\pi r \frac{\kappa a}{l^2} \phi \left( \frac{\phi_u}{r_u} \right)^2 - \frac{2 \kappa a \phi }{l^2 r}\nonumber \\
&&\hbox{}+ \frac{\phi_u}{r_u} \left( \frac{2 r_v}{r^2} - \frac{1}{r r_u} \partial_u \left(r \frac{ r_{uv}}{r_u} \right)  \right), \nonumber \\
&=& - \rho_0 A_{0}+ B_{0},
\end{eqnarray}
with $\rho_0= -4 \kappa \left( \frac{\varpi }{r^2} + \frac{r}{l^2}- \frac{4 \pi a \phi^2}{l^2} + \frac{1-\mu }{4r} \right)$ and \begin{eqnarray}
B_0&=& \frac{2}{r^2} \phi_v + \frac{2a \kappa}{l^2} \frac{\phi_u}{r_u} 
- 8\pi r \frac{\kappa a}{l^2} \phi \left( \frac{\phi_u}{r_u} \right)^2 - \frac{2 \kappa a \phi }{l^2 r}\nonumber \\
&&\hbox{}+ \frac{\phi_u}{r_u} \left( \frac{2 r_v}{r^2} - \frac{1}{r r_u} \partial_u \left( r \frac{r_{uv}}{r_u} \right)  \right). \nonumber
\end{eqnarray}
Let us define $A_{n}= r^n A_0$, $B_n= r^n B_0$. In this notation, we are looking for uniform bounds on $A_{7/2}(u,v)$.
We have:
\begin{eqnarray} \label{eq:anv}
\partial_v A_{n} = -\rho_{n} A_{n}+ B_{n} ,
\end{eqnarray}
with $\rho_n= \rho_0 -\frac{n}{r}\kappa (1-\mu)$. In view of the bounds already derived on our solution, we easily see
$$
\rho_n \ge 0
$$ for all $n \le 4$. Integrating \eqref{eq:anv}, we obtain, for $n=4$:
$$
| A_4 | (v) \le C \left( A_4(u_0) + \int^v_{u_0} r^{1/2} \left( r^{3/2} |\phi_{v}| +r^{3/2} |\phi_{u}| + r^{5/2} |\phi| \right) dv'\right),
$$
for some uniform constant $C$. Using Cauchy-Schwarz and the energy estimates to control the $L^2$ integral on $\phi_v$, $\phi$ and $\phi_u= T(\phi)-\phi_v$, one has:
$$
| r^{1/2}(u,v) A_{7/2} | \le C  \left(  r^{1/2}(u,u_0) A_{7/2}(u_0) + r^{1/2}(u,v) N(D) \right).
$$
Since $r(u,v) \ge r(u,u_0)$, we obtain a uniform bound on $A_{7/2}$ concluding the proof. Note that the proof only requires bounds for first derivatives of $\phi$.
\end{proof}

\section{Applications}
\subsection{Geometric uniqueness} \label{subse:md}
A well-known problem concerning geometric initial-value-boundary problems is the issue of geometric uniqueness. To prescribe the boundary conditions, one typically first needs to make a choice of gauge. The solution which is constructed may then a priori depend on this choice. See \cite{Friedrichbv} for a review of this problem for the Einstein equations. In our context, one may expect geometric uniqueness from \cite{FriedrichAdS}. In fact, since the boundary conditions are all stated with respect to geometric quantities, namely $\phi$ and $\varpi$ and the area radius, the only choice that we have made a priori is that of the $(u,v)$-coordinates.  This observation leads to a proof of the geometric uniqueness statement of Corollary \ref{cor:uniqueness}, which we present in this section.

First, we introduce the following definition:

\begin{definition} \label{def:dev}
Let $\mathcal{N}$ be an interval of the form $\mathcal{N}=(u_0,u_1]$. Given a $\mathcal{C}^{1+k}_{a,M}(\mathcal{N})$ asymptotically AdS data set $\mathcal{D}=(\bar{r}, \bar{\phi})$, a development of $\mathcal{D}$ is a triple $(\mathcal{M},g,\phi)$ such that $(\mathcal{M},g)$ is a $C^1$, $3+1$ Lorentzian manifold, $\phi$ is a $C^1$ function on $\mathcal{M}$ and the following hold
\begin{enumerate}
\item{$(\mathcal{M},g,\phi)$ is a spherically-symmetric solution to the Einstein-Klein-Gordon system \eqref{cons1}-\eqref{laste}, with area-radius $r$ being a $C^2$ function such that\footnote{We are excluding here the case where $r=0$ on some axis of symmetry for simplicity. This could nonetheless be handled as in the asymptotically flat case.} $r>0$.}
\item{The quotient manifold $\mathcal{Q}=\mathcal{M}/\mathcal{S}^2$ with its induced Lorentzian metric is a manifold with boundary $\mathcal{N}_{\mathcal{Q}}$ which is a null ray diffeomorphic to a subset of $\mathcal{N}$ of the form $(u_0,u_0+\epsilon)$, for some $\epsilon > 0$. If $\psi$ is such a diffeomorphism, $\psi: \mathcal{N}_{\mathcal{Q}} \rightarrow (u_0,u_0+\epsilon)$, then $\phi \circ \psi=\bar{\phi}|_{(u_0,u_0+\epsilon)}$ and $r \circ \psi=\bar{r}|_{(u_0,u_0+\epsilon)}$. \label{def:devdata}}
\item 
$\mathcal{Q}$ admits a system of global bounded null coordinates $(u,v)$ and hence may be conformally embedded into a bounded subset $\mathbb{R}^{1,1}$. The boundary of $\mathcal{Q}$ with respect to the topology of $\mathbb{R}^{1,1}$ is composed of a future boundary\footnote{Recall that the future boundary of $\mathcal{Q}$ is the set of points $p \in \partial \mathcal{Q}$, such that there exists no $q \in \mathcal{Q}$ such that $p \in J^-(q)$, where $J^-$ and $\partial \mathcal{Q}$ refer to the causal relations and the topology of $\mathbb{R}^{1+1}$.}, a past boundary which coincides with $\overline{\mathcal{N}_{Q}}$ and a $C^2$ timelike boundary $\mathcal{I}$.  As the boundary $\mathcal{I}$ is approached, the asymptotics of Section \ref{se:aads} hold for the metric, i.e. $(\mathcal{M},g)$ is asymptotically AdS in the sense of \cite{Holzegelwp}.
\item The following modified notion of global hyperbolicity holds: all past directed inextendible causal curves in $\mathcal{Q}$ either intersect $\mathcal{N}_{Q}$ or have limit endpoint on $\mathcal{I}$.
\item $\mathcal{T} \phi = -\frac{1-\mu}{r_u} \partial_u \phi + \frac{1-\mu}{r_v} \partial_v \phi \in C^0_u(H^1_{v,loc})\cap C^0_v(H^1_{u,loc})$, i.e. $\mathcal{T}\phi$ is continuous as a function of $u$ (respectively $v$) with image in the set of $H^1_{loc}$ function of $v$ (respectively $u$). Moreover, the field $\phi$ satisfies the following integrability conditions. 
For each constant $v$-ray, $R_v$ in $\mathcal{Q}$ we have 
\begin{align}
\int_{R_v} r^2 \left[  \frac{r^2}{|r_u|} \phi_u^2+ |r_u|\phi^2 + \frac{r^2}{|r_u|} \left(\partial_u \mathcal{T}\phi\right)^2+ |r_u|\left(\mathcal{T} \phi\right) ^2 \right]d\bar{u} < \infty \nonumber
\end{align}
and for each constant $u$-ray, $R_u$ in $\mathcal{Q}$ we have
\begin{align}
\int_{R_v} r^2 \left[  \frac{|1-\mu|r^2}{|r_v|} \phi_v^2+ |r_v|\phi^2 + \frac{|1-\mu|r^2}{|r_v|} \left(\partial_v \mathcal{T}\phi\right)^2+ |r_v|\left(\mathcal{T} \phi\right) ^2 \right]d\bar{v} < \infty. \nonumber
\end{align}

\item The renormalized Hawking mass $\varpi$ converges to $M$ on $\mathcal{I}$.
\end{enumerate}
\end{definition}
Hence, it follows from the above definition that the Penrose diagram\footnote{The reader not familiar with the usage of such diagrams will find a good introduction in Appendix $C$ of \cite{DafRod}. In particular, our conventions agree with that of \cite{DafRod} and \cite{CamJon}.} of a development of (characteristic) asymptotically AdS data takes the following form:

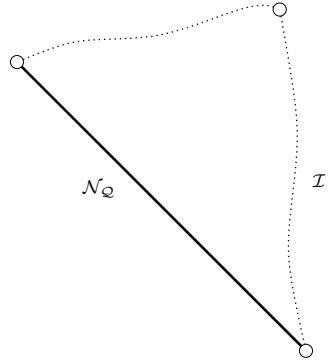
\[
\begin{picture}(0,0)%
\includegraphics{pendiadev.pstex}%
\end{picture}%
\setlength{\unitlength}{2171sp}%
\begingroup\makeatletter\ifx\SetFigFontNFSS\undefined%
\gdef\SetFigFontNFSS#1#2#3#4#5{%
  \reset@font\fontsize{#1}{#2pt}%
  \fontfamily{#3}\fontseries{#4}\fontshape{#5}%
  \selectfont}%
\fi\endgroup%
\begin{picture}(3473,4066)(4718,-5844)
\put(8176,-3886){\makebox(0,0)[lb]{\smash{{\SetFigFontNFSS{7}{8.4}{\rmdefault}{\mddefault}{\updefault}{\color[rgb]{0,0,0}$\mathcal{I}$}%
}}}}
\put(5551,-3961){\makebox(0,0)[lb]{\smash{{\SetFigFontNFSS{7}{8.4}{\rmdefault}{\mddefault}{\updefault}{\color[rgb]{0,0,0}$\mathcal{N}_\mathcal{Q}$}%
}}}}
\end{picture}%

\]

Given two developments $(\mathcal{M}_i,g_i,\phi_i)$, $i=1,2$, of a $\mathcal{C}^{1+k}_{a,M}(\mathcal{N})$ asymptotically AdS data set, we say that $(\mathcal{M}_1,g_1,\phi_1)$ is an extension of  $(\mathcal{M}_2,g_2,\phi_2)$ if there exists an isometric embedding $\psi$ of $(\mathcal{M}_2,g_2)$ into  $(\mathcal{M}_1,g_1)$, which maps $\phi_2$ to  $\phi_1$ and such that $\psi_1^{-1} \circ \psi \circ \psi_2$, is the identity on $\mathcal{N}$, where $\psi_i$, $i=1,2$ are the diffeormorphisms mapping $\mathcal{N}_{\mathcal{Q}_i}$ to $\mathcal{N}$ as in the above definition.

This definition makes the set of developments a partially ordered set. 
The maximal development is then by definition a maximal element, i.e.~a development which does not admit any extension. 

We now turn the proof of Corollary \ref{cor:uniqueness}. 
\begin{proof}We shall prove any development must agree, at least locally, with the development given by Theorem \ref{wellposed}. Given two developments, it then follows that they agree locally with the development given by Theorem \ref{wellposed} and hence are extension of a commom development. 

Let us thus be given a development $(\mathcal{M},g,\phi)$ of a $\mathcal{C}^{1+k}_{a,M}(\mathcal{N})$ asymptotically AdS data set, with $\mathcal{N}=(u_0,u_1]$. Let $\mathcal{N}_{Q}$ be the initial null ray as in Definition \ref{def:dev}. By part \ref{def:devdata} of the above definition, we may apply a change of $u$-coordinate so that $\mathcal{N}_{Q}$ is identified with a subset $(u_0,u_0+\epsilon)$ of $\mathcal{N}$. In this coordinate system, we have $r_u=\bar{r}_u$, $\phi_u=\bar{\phi}_u$, etc. on $\mathcal{N}_{Q}$. Since $\mathcal{I}$ is a $C^2$ timelike boundary, there exists a function $f$ with $f'>0$ such that $u=f(v)$ on $\mathcal{I}$. Let us consider the $v$ change of coordinates $V=f(v)$. Since this is $C^2$-coordinate transformation, it does not affect the finiteness of the invariant norms of $\phi$ and the metric is at least as regular as in Definition \ref{def:dev} in the new coordinate system. Hence, the variables $(r,\phi,\varpi,\Omega)$ of the development (and the associated renormalized variables), satisfy the Einstein-Klein-Gordon system in the coordinate system used in Proposition \ref{renowp} and have the regularity and boundary conditions of Proposition \ref{renowp}. By uniqueness, the solution must agree in the intersection of their domain of definitions, which contains in particular a neighborhood of $\mathcal{N} \cap \mathcal{I}$. 
\end{proof}
\subsection{An extension principle near infinity}
Assuming that certain uniform bounds hold in a triangle with the boundary being $\mathcal{I}$, we conclude that the solution can be extended to a larger triangle:
\begin{proposition} \label{pro:epin}
Let $\Psi=\left(r,\phi,\Omega,\varpi\right)$ be a solution of the EKG system in a triangular region $\Delta_{d,u_0}$ as arising by Theorem \ref{wellposed} from $\mathcal{C}^{2}_{a,M}$ initial data. Assume
\begin{align}
\lim_{v \rightarrow u_0+d} r \left(u_0+d,v\right) = \infty \, .
\end{align}
Suppose that there exist constants $0<c<\frac{1}{l^2}$, $C>0$ such that
\begin{align} \label{zeb}
\min \left(\inf_{\Delta_{d,u_0}}\Big|\frac{1-\mu}{r^2}\Big|, \inf_{\Delta_{d,u_0}} \Big| \frac{1-\mu_M}{r^2} \Big| \right) > c
\end{align}
\begin{align} \label{fib}
\sup_{\Delta_{d,u_0}} \Big| r^3 \left(\frac{r_u}{1-\mu} +\frac{1}{2}\right) \Big| + \sup_{\Delta_{d,u_0}} \Big| r^2 \partial_u \left(\frac{r_u}{1-\mu}\right) \Big| < C \, ,
\end{align}
and such that for any constant $v$-ray, $R_v$, contained in $\Delta_{d,u_0}$ and intersecting $\mathcal{I}$
\begin{align} \label{seb}
 \int_{R_v} r^2 \left[  \frac{r^2}{|r_u|} \phi_u^2+ |r_u|\phi^2 + \frac{r^2}{|r_u|} \left(\partial_u \mathcal{T}\phi\right)^2+ |r_u|\left(\mathcal{T} \phi\right) ^2 \right]du' + \Big| r^{\frac{5+s}{2}} \frac{\phi_u}{r_u} \Big|  < C \, .
\end{align}
Then there exists a $\delta^\star>0$ such that the solution $\Psi$ can be extended to the strictly larger triangle $\Delta_{d+\delta^\star,u_0}$.
\end{proposition}
\begin{proof}
Clearly, the solution can be extended to the set $\Delta_{d+\tilde{\delta},u_0} \cap \{v \leq u_0+d+\tilde{\delta} -\epsilon \}$ for some $\tilde{\delta}>0$ which depends on $\epsilon$ by continuity and the standard local well-posedness result available away from the boundary.

To extend it to a full triangle $\Delta_{d+\delta^\star,u_0}$, note that we have a uniformly bounded asymptotically AdS initial-data set on each $v=const$-ray in $\Delta_{d,u_0}$. Let $\delta$ be the time of existence ($v$-length) associated with the bounds (\ref{zeb}), (\ref{fib}) and (\ref{seb}) by Theorem \ref{wellposed}, however with $C$ replaced by $2C$ and $c$ by $\frac{c}{2}$ respectively in(\ref{zeb})-(\ref{seb}). Pick the ray $v_c = u_0 + d - \frac{\delta}{2}$.
By the above argument (and continuity), we can extend the solution to the ray $(u_0+d-\frac{\delta}{2}, u_0+d + \delta^\star] \times \{ u_0+d-\frac{\delta}{2} \}$ for some $\delta^\star < \frac{\delta}{2}$ such that moreover  (\ref{zeb})-(\ref{seb}) hold on $(u_0+d-\frac{\delta}{2}, u_0+d + \delta^\star] \times \{ u_0+d-\frac{\delta}{2} \}$ with the constant $C$ replaced by $2C$ and $c$ by $\frac{c}{2}$ respectively. Applying Theorem \ref{wellposed} extends the solution to all of $\Delta_{d+\delta^\star,u_0}$.
\end{proof}

\begin{appendix}
\section{An extension principle in the interior}
In this appendix, we present a second extension principle, which regards the properties of solutions in the interior (that is, away from $\mathcal{I}$) of the spacetime. Remarkably, this extension principle does not use the energy conservation and (as a consequence) is applicable also in the interior of the black hole. We thank Mihalis Dafermos and Jonathan Kommemi for introducing us to the argument presented here.
\begin{proposition} \label{sexp}
Let $\left(\mathcal{Q}^+ \times S^2, g, \phi\right)$ denote the maximum development of an asymptotically AdS initial data set. Suppose $p=\left(U,V\right) \in \overline{\mathcal{Q}^+}$. If 
\begin{enumerate}
\item $\mathcal{D}= \left[U^\prime,U\right] \times \left[V^\prime,V\right] \setminus \{p \} \subset \mathcal{Q}^+$ has finite spacetime volume 
\item there exist constants $r_0$ and $R$ such that
\begin{align}
0 < r_0 \leq r \left(u,v\right) \leq R < \infty \textrm{ \ \ \ for all $\left(u,v\right) \in \mathcal{D}$,}
\end{align}
\end{enumerate}
Then $p \in \mathcal{Q}^+$.
\end{proposition}
\begin{proof}
We first note that since $r_u<0$ holds near $\mathcal{I}$ for an asymptotically AdS spacetime, we must have $r_u<0$ in the entire maximal development as a consequence of the Raychaudhuri equation (\ref{kappae}). By the assumptions 1 and 2 we have
\begin{align} \label{fivol}
\int_{U^\prime}^U \int_{V^\prime}^V  \Omega^2 dU dV < C 
\end{align}
and
\begin{align}
\frac{1}{C} < r_0 \leq r \left(u,v\right) \leq R < C \, ,
\end{align}
for some constant $C$. Moreover, by compactness, we have on $\left[U^\prime,U\right] \times \{V^\prime\}$ and $\{ U^\prime \} \times \left[V^\prime,V\right]$ the estimates
\begin{align}
\frac{1}{N} < -r \cdot r_u < N   \textrm{ \ \ \ , \ \ \ } | r \cdot r_v| < N \nonumber \\
| r \phi | + |r \phi_v| + |r \phi_u| + |\phi_{uu}| + |\phi_{uv} | + \int_{V^\prime}^V \left(\partial_v T\left(\phi\right)\right)^2 \left(U^\prime,v\right) dv < N  \nonumber \\
|r_{uv}| + | r_{uu} | + |r_{vv}| + |\partial_u \Omega| + | \partial_v \Omega| + | \log \Omega^2 | < N
\end{align}
for some constant $N$. We write (\ref{eq:ruv}) in the form
\begin{align}
\partial_u \left(r \lambda\right) = -\frac{\Omega^2}{4} - \frac{3}{4} \frac{r^2}{l^2} \Omega^2 + \frac{2\pi r^2 a \Omega^2}{l^2} \phi^2 \,, 
\end{align}
where $\lambda=r_v$.
Using the bounds on the spacetime volume and the area radius one derives the spacetime bound
\begin{align}
\int_{U^\prime}^U \int_{V^\prime}^V  \Omega^2 \phi^2 dU dV < \tilde{C} \, ,
\end{align}
where $\tilde{C}$ depends only on $l$ and $C$. We also note the pointwise estimate
\begin{align} \label{pwrl}
 \sup_{\left[U^\prime,U\right]} | r\lambda | \leq N - \frac{1}{4} \int_{U^\prime}^U dU \left(1+ 3 \frac{r^2}{l^2}\right) \Omega^2 + \int_{U^\prime}^U dU 2\pi r^2 \frac{a}{l^2} \Omega^2 \phi^2 \, ,
\end{align}
which upon integration yields
\begin{align}  \label{rl}
\int_{V^\prime}^V dV \sup_{\left[U^\prime,U\right]} | r\lambda | \leq N \left( V- V^\prime\right) + \tilde{C} \, ,
\end{align}
and similarly
\begin{align}
\int_{U^\prime}^U dU \sup_{\left[V^\prime,V\right]} | r\nu | \leq N \left( U- U^\prime\right) + \tilde{C} \, , 
\end{align}
where $\nu=r_u$.
We now partition the diamond $\mathcal{D}$ into smaller sub-diamonds $\mathcal{D}_{jk}$ given by
\begin{align}
\mathcal{D}_{jk} = \left[u_j, u_{j+1}\right] \times \left[v_j, v_{j+1}\right]  \textrm{ \ \ \ \ \ \  $j,k = 0, ... , N$}
\end{align}
and with $u_0=U^\prime$, $u_N=U$, $v_0=V^\prime$, $v_N=V$ and such that for a given $\epsilon>0$ we have
\begin{align}
\int_{v_k}^{v_{k+1}} \int_{u_j}^{u_{j+1}}  \Omega^2 du dv < \epsilon
\textrm{ \ \ \  \ \  and  \ \ \ \ }
\int_{v_k}^{v_{k+1}} \sup_{[u_j,u_{j+1}]} |r \lambda | dv < \epsilon \, .
\end{align}
This is possible in view of the uniform bounds (\ref{fivol}) and (\ref{rl}). Define also
\begin{align}
P_{jk} = \sup_{\mathcal{D}_{jk}} | r \phi \left(u,v\right)| \, .
\end{align}
Pick an arbitrary point $\left(u^\star,v^\star\right) \in \mathcal{D}_{jk}$ and consider the wave equation for $\phi$:
\begin{align} \label{wef}
\partial_u \partial_v \left(r\phi\right) = \phi \partial_u \lambda - \frac{a r}{2l^2} \Omega^2 \phi \, .
\end{align}
We have
\begin{align}
\int_{v_k}^{v^\star}\int_{u_j}^{u^\star}  \frac{a r}{2l^2} \Omega^2 \phi du dv  \leq  C_{a,l} P_{jk} \cdot \epsilon \, ,
\end{align}
and also
\begin{align}
\int_{v_k}^{v^\star}\int_{u_j}^{u^\star} \phi \partial_u \lambda du dv \nonumber \\
= \int_{v_k}^{v^\star}\int_{u_j}^{u^\star} \phi  \left( - \frac{3r \Omega^2}{4l^2} - \frac{\Omega^2}{4r} - \frac{r_u r_v}{r} + 2\pi r \frac{a}{l^2} \Omega^2 \phi ^2 \right) du dv 
\leq P_{jk} \cdot \frac{C_l}{r_0^2} \cdot \epsilon \nonumber
\end{align}
in view of
\begin{align}
\int_{v_k}^{v^\star}\int_{u_j}^{u^\star} - \phi \frac{r_u r_v}{r}  du dv  \leq P_{jk} \int_{v_k}^{v^\star}\int_{u_j}^{u^\star} \frac{-r_u}{r^3} | r \lambda|  du dv \nonumber \\
 \leq P_{jk} \int_{v_k}^{v^\star}  \sup_{[u_j,u_{j+1}]} |r \lambda | dv  \int_{u_j}^{u^\star} \frac{-r_u}{r^3} du \leq P_{jk} \cdot \frac{C_l}{r_0^2} \cdot \epsilon \, . \nonumber 
\end{align}
Hence integrating (\ref{wef}) in $u$ and $v$ yields for sufficiently small $\epsilon$ the uniform bound
\begin{align}
P_{jk} < 2 \left(\sup_{[u_j,u_{j+1}] \times \{v_k\}} | r\phi| + \sup_{\{u_j\} \times [v_k,v_{k+1}] } | r\phi| \right)< 2 \left( P_{j,k-1} + P_{j-1,k}\right)
\end{align}
Inductively one steps back to $P_{0,k}$ and $P_{j,0}$ therefore obtaining a uniform bound on $P_{jk}$ in terms of the initial data. Taking the maximum over all sub-diamonds yields
\begin{align}
 \sup_{\mathcal{D}} | r \phi \left(u,v\right)| < \tilde{C} \, .
\end{align}
We continue by proving a pointwise bound on $\log \Omega^2$. In view of the evolution equation for this quantity, all this requires is the spacetime bound
\begin{align} 
\int_{U^\prime}^U \int_{V^\prime}^V  \partial_u \phi \partial_v \phi \ dU dV \nonumber \\
= \int_{U^\prime}^U \int_{V^\prime}^V  \frac{1}{2} \partial_u \partial_v \left(\phi^2\right)  - \phi \left(-\frac{r_u}{r} \phi_v - \frac{r_v}{r} \phi_u - \frac{a}{2l^2} \Omega^2 \phi \right) \ dU dV \nonumber \\
= \int_{U^\prime}^U \int_{V^\prime}^V  \left[ \frac{1}{2} \partial_u \partial_v \left(\phi^2\right) + \frac{r_u}{2r} \partial_v \left(\phi^2\right) +   \frac{r_v}{2r} \partial_u \left(\phi^2\right)  + \frac{a}{2l^2} \Omega^2 \phi^2 \right]  dU dV 
< \tilde{C} \nonumber \, .
\end{align}
The latter bound is immediate for the first and the last term in the square bracket in view of previous bounds on $\phi$. The remaining terms can be integrated by parts and (using the evolution equation for $r_{uv}$) are easily seen to be bounded. We conclude
\begin{align}
|\log \Omega^2 | < \tilde{C}
\end{align}
With the pointwise bound on $\Omega^2$ and $\phi$ we now control $T_{uv}$ itself pointwise. Revisiting (\ref{pwrl}) therefore yields
\begin{align}
\sup_{\mathcal{D}} | r \lambda | + \sup_{\mathcal{D}} | r \nu |< \tilde{C} \, .
\end{align}
The bounds for higher derivatives are then straightforward using the evolution equations. The proposition finally follows as in \cite{CamJon} by applying a standard existence result sufficiently close to $p$, in view of the uniform bounds just derived.
\end{proof}

Starting from Proposition \ref{sexp} one can repeat the analysis of \cite{CamJon} and determine the general global structure of the maximum development of spherically-symmetric EKG spacetimes. In the asymptotically flat case, it is shown in \cite{CamJon} (for the Einstein-Maxwell-charged Klein-Gordon equations, which contains in particular the Einstein-Klein-Gordon system) that the Penrose diagram of the evolution of any data set with a single asymptotically-flat end is as follows:


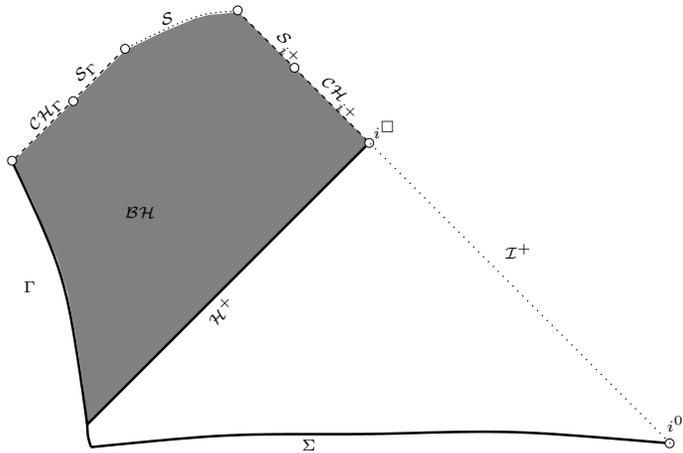
\[
\begin{picture}(0,0)%
\includegraphics{pengep.pstex}%
\end{picture}%
\setlength{\unitlength}{2072sp}%
\begingroup\makeatletter\ifx\SetFigFontNFSS\undefined%
\gdef\SetFigFontNFSS#1#2#3#4#5{%
  \reset@font\fontsize{#1}{#2pt}%
  \fontfamily{#3}\fontseries{#4}\fontshape{#5}%
  \selectfont}%
\fi\endgroup%
\begin{picture}(7989,5408)(1965,-4949)
\put(5491,-4876){\makebox(0,0)[lb]{\smash{{\SetFigFontNFSS{6}{7.2}{\rmdefault}{\mddefault}{\updefault}{\color[rgb]{0,0,0}$\Sigma$}%
}}}}
\put(9856,-4651){\makebox(0,0)[lb]{\smash{{\SetFigFontNFSS{6}{7.2}{\rmdefault}{\mddefault}{\updefault}{\color[rgb]{0,0,0}$i^0$}%
}}}}
\put(3376,-2086){\makebox(0,0)[lb]{\smash{{\SetFigFontNFSS{6}{7.2}{\rmdefault}{\mddefault}{\updefault}{\color[rgb]{0,0,0}$\mathcal{BH}$}%
}}}}
\put(4456,-3391){\rotatebox{45.0}{\makebox(0,0)[lb]{\smash{{\SetFigFontNFSS{6}{7.2}{\rmdefault}{\mddefault}{\updefault}{\color[rgb]{0,0,0}$\mathcal{H}^+$}%
}}}}}
\put(3826,209){\rotatebox{20.0}{\makebox(0,0)[lb]{\smash{{\SetFigFontNFSS{6}{7.2}{\rmdefault}{\mddefault}{\updefault}{\color[rgb]{0,0,0}$\mathcal{S}$}%
}}}}}
\put(6346,-1141){\makebox(0,0)[lb]{\smash{{\SetFigFontNFSS{6}{7.2}{\rmdefault}{\mddefault}{\updefault}{\color[rgb]{0,0,0}$i^{\square}$}%
}}}}
\put(5716,-466){\rotatebox{315.0}{\makebox(0,0)[lb]{\smash{{\SetFigFontNFSS{6}{7.2}{\rmdefault}{\mddefault}{\updefault}{\color[rgb]{0,0,0}$\mathcal{CH}_{i^+}$}%
}}}}}
\put(5176, 74){\rotatebox{315.0}{\makebox(0,0)[lb]{\smash{{\SetFigFontNFSS{6}{7.2}{\rmdefault}{\mddefault}{\updefault}{\color[rgb]{0,0,0}$\mathcal{S}_{i^+}$}%
}}}}}
\put(2296,-1051){\rotatebox{45.0}{\makebox(0,0)[lb]{\smash{{\SetFigFontNFSS{6}{7.2}{\rmdefault}{\mddefault}{\updefault}{\color[rgb]{0,0,0}$\mathcal{CH}_{\Gamma}$}%
}}}}}
\put(2836,-466){\rotatebox{45.0}{\makebox(0,0)[lb]{\smash{{\SetFigFontNFSS{6}{7.2}{\rmdefault}{\mddefault}{\updefault}{\color[rgb]{0,0,0}$\mathcal{S}_{\Gamma}$}%
}}}}}
\put(2161,-2986){\makebox(0,0)[lb]{\smash{{\SetFigFontNFSS{6}{7.2}{\rmdefault}{\mddefault}{\updefault}{\color[rgb]{0,0,0}$\Gamma$}%
}}}}
\put(7921,-2581){\makebox(0,0)[lb]{\smash{{\SetFigFontNFSS{6}{7.2}{\rmdefault}{\mddefault}{\updefault}{\color[rgb]{0,0,0}$\mathcal{I}^+$}%
}}}}
\end{picture}%

\]
We refer to \cite{CamJon} for a precise description of all the components of the boundary of $\mathcal{Q}$. 

It is instructive to relate the results of \cite{CamJon} to the asymptotically AdS case considered in this paper. We remark that because the main theorem of \cite{CamJon} only uses the monotonicity of the Raychaudhuri equations (but not the monotonicity of the Hawking mass!) and the extension principle of proposition \ref{sexp}, a very similar picture can be established for the AdS case. More precisely, the quotient of the maximal development of any asymptotically AdS initial data set with one end has the following Penrose diagram:

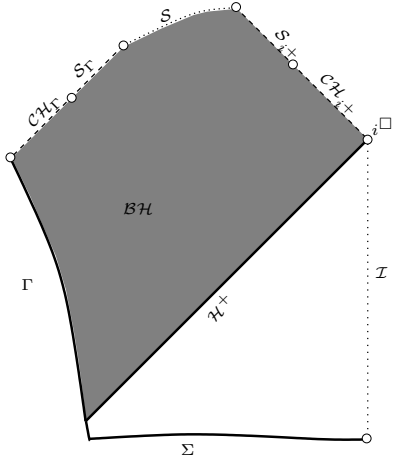
\[
\begin{picture}(0,0)%
\includegraphics{pengepads.pstex}%
\end{picture}%
\setlength{\unitlength}{2072sp}%
\begingroup\makeatletter\ifx\SetFigFontNFSS\undefined%
\gdef\SetFigFontNFSS#1#2#3#4#5{%
  \reset@font\fontsize{#1}{#2pt}%
  \fontfamily{#3}\fontseries{#4}\fontshape{#5}%
  \selectfont}%
\fi\endgroup%
\begin{picture}(4441,5498)(1965,-5039)
\put(4065,-4966){\makebox(0,0)[lb]{\smash{{\SetFigFontNFSS{6}{7.2}{\rmdefault}{\mddefault}{\updefault}{\color[rgb]{0,0,0}$\Sigma$}%
}}}}
\put(3376,-2086){\makebox(0,0)[lb]{\smash{{\SetFigFontNFSS{6}{7.2}{\rmdefault}{\mddefault}{\updefault}{\color[rgb]{0,0,0}$\mathcal{BH}$}%
}}}}
\put(4456,-3391){\rotatebox{45.0}{\makebox(0,0)[lb]{\smash{{\SetFigFontNFSS{6}{7.2}{\rmdefault}{\mddefault}{\updefault}{\color[rgb]{0,0,0}$\mathcal{H}^+$}%
}}}}}
\put(3826,209){\rotatebox{20.0}{\makebox(0,0)[lb]{\smash{{\SetFigFontNFSS{6}{7.2}{\rmdefault}{\mddefault}{\updefault}{\color[rgb]{0,0,0}$\mathcal{S}$}%
}}}}}
\put(6346,-1141){\makebox(0,0)[lb]{\smash{{\SetFigFontNFSS{6}{7.2}{\rmdefault}{\mddefault}{\updefault}{\color[rgb]{0,0,0}$i^{\square}$}%
}}}}
\put(5716,-466){\rotatebox{315.0}{\makebox(0,0)[lb]{\smash{{\SetFigFontNFSS{6}{7.2}{\rmdefault}{\mddefault}{\updefault}{\color[rgb]{0,0,0}$\mathcal{CH}_{i^+}$}%
}}}}}
\put(5176, 74){\rotatebox{315.0}{\makebox(0,0)[lb]{\smash{{\SetFigFontNFSS{6}{7.2}{\rmdefault}{\mddefault}{\updefault}{\color[rgb]{0,0,0}$\mathcal{S}_{i^+}$}%
}}}}}
\put(2296,-1051){\rotatebox{45.0}{\makebox(0,0)[lb]{\smash{{\SetFigFontNFSS{6}{7.2}{\rmdefault}{\mddefault}{\updefault}{\color[rgb]{0,0,0}$\mathcal{CH}_{\Gamma}$}%
}}}}}
\put(2836,-466){\rotatebox{45.0}{\makebox(0,0)[lb]{\smash{{\SetFigFontNFSS{6}{7.2}{\rmdefault}{\mddefault}{\updefault}{\color[rgb]{0,0,0}$\mathcal{S}_{\Gamma}$}%
}}}}}
\put(2161,-2986){\makebox(0,0)[lb]{\smash{{\SetFigFontNFSS{6}{7.2}{\rmdefault}{\mddefault}{\updefault}{\color[rgb]{0,0,0}$\Gamma$}%
}}}}
\put(6391,-2851){\makebox(0,0)[lb]{\smash{{\SetFigFontNFSS{6}{7.2}{\rmdefault}{\mddefault}{\updefault}{\color[rgb]{0,0,0}$\mathcal{I}$}%
}}}}
\end{picture}%

\]

%


An important difference between the two cases regards the completeness of null infinity (i.e.~when $i^\square = i^+$ in the notation of \cite{CamJon}). In the asymptotically flat case, the standard proof of completeness of future null infinity, $\mathcal{I}^+$, relies on the monotonicity properties of the Hawking mass. In the AdS setting, this monotonicity is not available, hence the proof of the completeness of null infinity requires a different analysis. In our companion paper \cite{gs:stab}, we prove in particular the completeness of $\mathcal{I}$ for perturbations of Schwarzschild-AdS initial data.


The following Corollary is an immediate consequence of the global spacetime structure established above. However, since it is precisely this statement which will be applied in our companion \cite{gs:stab}, we give an explicit proof.
\begin{corollary}
Let $\left(\bar{r},\bar{\phi}\right)$ be an $\mathcal{C}^{1+k}_{a,M}\left(\mathcal{N}\right)$ asymptotically AdS data set with $k\geq1$, which contains a (marginally) trapped surface, i.e.~a point on $\mathcal{N}$ for which $r_v\leq 0$. Then the quotient of the maximum development of the data set must necessarily contain a subset as depicted below
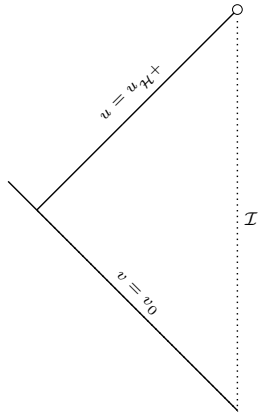
\[
\begin{picture}(0,0)%
\includegraphics{penAdS3.pstex}%
\end{picture}%
\setlength{\unitlength}{2368sp}%
\begingroup\makeatletter\ifx\SetFigFontNFSS\undefined%
\gdef\SetFigFontNFSS#1#2#3#4#5{%
  \reset@font\fontsize{#1}{#2pt}%
  \fontfamily{#3}\fontseries{#4}\fontshape{#5}%
  \selectfont}%
\fi\endgroup%
\begin{picture}(2512,4281)(579,-3683)
\put(3076,-1711){\makebox(0,0)[lb]{\smash{{\SetFigFontNFSS{7}{8.4}{\rmdefault}{\mddefault}{\updefault}{\color[rgb]{0,0,0}$\mathcal{I}$}%
}}}}
\put(1726,-2236){\rotatebox{315.0}{\makebox(0,0)[lb]{\smash{{\SetFigFontNFSS{6}{7.2}{\rmdefault}{\mddefault}{\updefault}{\color[rgb]{0,0,0}$v=v_0$}%
}}}}}
\put(1604,-621){\rotatebox{45.0}{\makebox(0,0)[lb]{\smash{{\SetFigFontNFSS{6}{7.2}{\rmdefault}{\mddefault}{\updefault}{\color[rgb]{0,0,0}$u=u_{\mathcal{H}^+}$}%
}}}}}
\end{picture}%

\]
where $u_{\mathcal{H}}$ is the boundary of the region for which $r = \infty$ can be reached along constant $u$-rays. Moreover, the set $u=u_{\mathcal{H}^+}$ belongs to the quotient of the maximal development. 
\end{corollary}
\begin{proof}
By Theorem \ref{wellposed} there exists a solution in a small triangle with $r \rightarrow \infty$ along any constant $u$ ray. Since by the Raychaudhuri equation (\ref{cons2}) spacetime points for which $r_v \leq 0$ cannot reach $r\rightarrow \infty$ (i.e.~$\mathcal{I}$) in their future,  there exists a maximal $u=u_{\mathcal{H}^+}$ such that $r \rightarrow \infty$ along all rays with $u<u_{\mathcal{H}^+}$. Finally, the ray $u=u_{\mathcal{H}^+}$ is regular, since first singularities along it are excluded by Proposition \ref{sexp} and the monotonicity of $r$.
\end{proof}

\end{appendix}

\bibliographystyle{hacm}
\bibliography{thesisrefs}

\begin{thebibliography}{10}

\bibitem{Bachelot}
{\sc Bachelot, A.}
\newblock {The Dirac System on the Anti-de Sitter Universe}.
\newblock {\em Commun.~Math.~Phys. 283\/} (2008), 127--167.

\bibitem{Bachelot2}
{\sc Bachelot, A.}
\newblock {The Klein-Gordon Equation in Anti-de Sitter Cosmology}.
\newblock arXiv:1010.1925.

\bibitem{Breitenlohner}
{\sc Breitenlohner, P., and Freedman, D.~Z.}
\newblock {Stability in Gauged Extended Supergravity}.
\newblock {\em Ann. Phys. 144\/} (1982), 249.

\bibitem{Geroch}
{\sc Choquet-Bruhat, Y., and Geroch, R.~P.}
\newblock {Global aspects of the Cauchy problem in General Relativity}.
\newblock {\em Comm.~Math.~Phys. 14\/} (1969), 329--335.

\bibitem{Christodoulou}
{\sc Christodoulou, D.}
\newblock {The Problem of a Selfgravitating Scalar Field}.
\newblock {\em Commun. Math. Phys. 105\/} (1986), 337--361.

\bibitem{Christodoulou3}
{\sc Christodoulou, D.}
\newblock {A mathematical theory of gravitational collapse}.
\newblock {\em Commun. Math. Phys. 109\/} (1987), 613--647.

\bibitem{Christodoulou4}
{\sc Christodoulou, D.}
\newblock {The instability of naked singularities in the gravitational collapse
  of a scalar field}.
\newblock {\em Ann. of Math. 149\/} (1999), 183--217.

\bibitem{Mihali1}
{\sc Dafermos, M.}
\newblock {Spherically symmetric spacetimes with a trapped surface}.
\newblock {\em Class. Quant. Grav. 22\/} (2005), 2221--2232, gr-qc/0403032.

\bibitem{DafRod}
{\sc Dafermos, M., and Rodnianski, I.}
\newblock {A proof of Price's law for the collapse of a self- gravitating
  scalar field}.
\newblock {\em Invent. Math. 162\/} (2005), 381--457, gr-qc/0309115.

\bibitem{FriedrichAdS}
{\sc Friedrich, H.}
\newblock {Einstein equations and conformal structure: existence of anti-de
  Sitter-type space-times}.
\newblock {\em J. Geom. Phys. 17\/} (1995), 125--184.

\bibitem{Friedrichbv}
{\sc Friedrich, H.}
\newblock {Initial boundary value problems for Einstein's field equations and
  geometric uniqueness}.
\newblock {\em Gen. Relativity Gravitation 41\/} (2009), 1947--1966.

\bibitem{Gubser}
{\sc Gubser, S.~S.}
\newblock {Breaking an Abelian gauge symmetry near a black hole horizon}.
\newblock {\em Phys. Rev. D78\/} (2008), arXiv:065034, 0801.2977.

\bibitem{HolzegelAdS}
{\sc Holzegel, G.}
\newblock {On the massive wave equation on slowly rotating Kerr-AdS
  spacetimes}.
\newblock {\em Comm.~Math.~Phys. 294\/} (2010), 169--197, arXiv:0902.0973.

\bibitem{Holzegelwp}
{\sc Holzegel, G.}
\newblock {Well-posedness for the massive wave equation on asymptotically
  anti-de-Sitter spacetimes}.
\newblock 20 pages, arXiv:1103.0710.

\bibitem{gs:stab}
{\sc Holzegel, G., and Smulevici, J.}
\newblock {Stability of Schwarzschild-AdS for the spherically symmetric
  Einstein-Klein Gordon system}.
\newblock arXiv:1103.3672.

\bibitem{CamJon}
{\sc Kommemi, J.}
\newblock {The global structure of spherically symmetric charged scalar field
  spacetimes}.
\newblock {\em preprint\/} (2011), arXiv:1107.0949.

\bibitem{Julianos}
{\sc Sonner, J.}
\newblock {A Rotating Holographic Superconductor}.
\newblock {\em Phys. Rev. D80\/} (2009), 084031, arXiv:0903.0627.

\bibitem{Vasy2}
{\sc Vasy, A.}
\newblock {The wave equation on asymptotically Anti-de Sitter spaces}.
\newblock {\em to appear in Analysis and PDE\/} (2009), arXiv:0911.5440.

\end{thebibliography}
\end{document}